\documentclass[10pt]{article}
\usepackage{amsmath,amssymb,amsthm}
\usepackage{mathtools}
\usepackage{graphicx}
\usepackage{natbib}
\usepackage{hyperref}
\usepackage{geometry}
\usepackage{setspace}
\usepackage{enumitem}
\usepackage{bm}
\usepackage{dsfont}           
\usepackage{xcolor}
\usepackage{caption}
\usepackage{comment}
\usepackage{subcaption}
\usepackage{booktabs}
\usepackage{array}
\usepackage[T1]{fontenc}
\usepackage{float}
\theoremstyle{plain}
\newtheorem{theorem}{Theorem}
\newtheorem{Result}{Result}

\newtheorem{lemma}{Lemma}

\newcommand{\Var}{\mathrm{Var}}
\newcommand{\Cov}{\mathrm{Cov}}

\title{The evolution of cooperation under imperfect phenotypic recognition}

\author{Dhaker Kroumi$^1$\footnote{Author for
correspondence: dhaker.kroumi@kfupm.edu.sa} 
\\$^1$Department of Mathematics and Statistics\\King Fahd University of Petroleum and Minerals\\Dhahran 31261, Saudi Arabia\\
 }
\date{}
\begin{document}
\maketitle


\begin{abstract}
Phenotypic similarity is a classical mechanism for the evolution of cooperation. Most existing models assume a binary rule in which individuals cooperate only with others of the same phenotype. This assumption is biologically restrictive, since recognition and discrimination are often gradual rather than all-or-nothing. In this paper, we extend the multidimensional phenotype-space model of cooperation by allowing the probability of helping to decline with phenotypic distance. In a large population under weak selection with mutation in both strategy and phenotype, we derive a generalized threshold for selection to favor the abundance of cooperation and express it using a new Laplace-type transform.
We show that this threshold decreases strictly as discrimination becomes sharper, meaning that exact phenotype matching is the most favorable limit within this family of recognition rules. We also show that the threshold decreases strictly with phenotype-space dimension, meaning that higher-dimensional phenotype spaces promote cooperation even when recognition is imperfect. Our asymptotic analysis further shows that low phenotype mutations strongly inhibit cooperation, whereas sufficiently high phenotype mutations drive the threshold toward its minimal value. However, high strategy mutation makes cooperation harder to maintain.
\end{abstract}


\section{Introduction}

The emergence of cooperation among selfish individuals remains one of the central questions in evolutionary game theory \citep{AxelrodHamilton1981,MaynardSmith1982}. In the Prisoner’s Dilemma, a cooperator pays a cost to provide a benefit to its partner, whereas a defector pays no cost and provides no benefit. In unstructured populations, natural selection favors defection even though cooperation can improve collective welfare \citep{Nowak2006}. Thus, the persistence of cooperation is often explained by mechanisms that generate positive assortment among cooperators, including kin selection \citep{Hamilton1963,Hamilton1964a,Hamilton1964b}, repeated interaction \citep{Trivers1971,Axelrod1984,NowakSigmund1998,BoydRicherson2005}, and population structure \citep{Wilson1987,NowakMay1992,Ohtsuki2006}. These mechanisms suggest that cooperative acts cannot be indiscriminate. Instead, they must be directed toward individuals who are likely to reciprocate or who share the same cooperative tendency.

Phenotypic similarity provides a natural basis for positive assortment. When cooperators direct help toward individuals with similar phenotypes, cooperation can persist through a mechanism that is closely related to the green-beard effect \citep{Dawkins1976,GardnerWest2010}. This idea has inspired a large body of work on tag-based cooperation, where social behavior depends on heritable visible markers rather than on direct information about genealogy \citep{Riolo2001,Axelrod2004,TraulsenNowak2007}. In these models, the tag or phenotype usually has no intrinsic fitness value. Its evolutionary role lies in the assortment created among individuals with similar behavioral tendencies. Related studies have also shown that similarity-based cooperation can produce rich dynamics, including segregation in phenotype space and continual turnover of cooperative types \citep{TraulsenClausen2004,Jansen2006}.

For finite populations, \citep{Antal2009c} have analyzed a model of a phenotype evolving according to a Wright-Fisher where each individual is characterized by a phenotype given by an integer and a strategy among $C$ for cooperation or $D$ for defection. Cooperators cooperate exclusively with others sharing their phenotypic marker, otherwise, they will defect. During reproduction, an offspring may undergo phenotypic mutation, under which  it shifts its phenotype by one step to the left or right. \citep{Kroumi2015} extended this framework to a phenotype space given by the $n$-dimensional lattice $\mathbb{Z}^n$, suggesting that increasing the number of phenotypic dimensions lowers the threshold and thereby facilitates the evolution of cooperation.

Within this framework, however, the interaction rule has relied on exact phenotype matching: a cooperator helps another individual if and only if both share the same phenotype. Although this assumption is analytically convenient, kin-recognition systems are often based on multiple cues and need not provide an exact readout of genealogical similarity \citep{Grafen1990,PennFrommen2010}. More generally, recognition can depend on learning, prior association, and environmental or social context, rather than on a fixed binary mechanism \citep{TangMartinez2001,Hepper2011}. Theoretical work has also emphasized that marker-based recognition is subject to important evolutionary constraints, which further suggests that exact and perfectly reliable discrimination should not be treated as the default case \citep{RoussetRoze2007}. Exact recognition is therefore better viewed as a useful limiting case than as the biologically typical scenario.

Empirical and theoretical work increasingly suggests that recognition in green-beard systems is not always strictly binary. In budding yeast, \citep{Smukalla2008} showed that the \textit{FLO1} gene drives cooperative flocculation through preferential self/non-self adhesion and that this trait is variable across strains, while later work showed that recognition error can affect its evolutionary stability \citep{Choi2023}. In social amoebae, \citep{Gruenheit2017} identified a polychromatic locus showing that partner-specific cooperation can depend on compatibility across multiple recognition types rather than on exact identity alone. At the conceptual level, \citep{GardnerWest2010,Madgwick2019} emphasize that real green-beard systems need not conform to a simple picture of perfectly binary recognition, and related theory on multicolored variants reaches a similar conclusion \citep{Jansen2006,BiernaskieEtAl2013}.

In this paper, we model graded discrimination using an exponential distance-decay kernel, a simple and analytically convenient choice within the broader family of distance-based kernels widely used in ecology and dispersal modeling \citep{KotLewisvanDenDriessche1996,Clark1999,NathanEtAl2012,RogersEtAl2019}. This form provides a natural way to represent interactions that weaken smoothly with phenotypic distance while remaining tractable in multidimensional space. We then derive an explicit threshold condition expressed through a generalized Laplace-type transform and study how it depends on discrimination sharpness, the dimension of the phenotype space, and mutation rates.

The rest of the paper is organized as follows. Section \ref{sec1} introduces the model and the graded phenotypic discrimination rule. Section \ref{sec2} derives the condition to favor the abundance of cooperation and expresses the resulting threshold in terms of a generalized Laplace-type transform. Section \ref{sec3} investigates the role of discrimination sharpness, and Section \ref{sec4} studies the dependence of this threshold on phenotype-space dimension. Section \ref{sec5} is devoted to the main asymptotic regimes of the mutation parameters, while Section \ref{sec6} extends the analysis to general symmetric two-player games. The technical derivations and proofs are given in Appendices \ref{AppendixA}--\ref{AppendixD}. We discuss the biological interpretation of the results and their relation to earlier work in Section \ref{sec8}. 

\section{\label{sec1}Model}
Consider a finite haploid population of size $N$ evolving according to a Wright–Fisher model. Number these individuals from $1$ to $N$.
Each individual $k\in\{1,\,2,\ldots,\,N\}$ is characterized by:
\begin{itemize}
\item \textbf{Social Strategy}: $S(k) \in \{C, D\}$, where $C$ and $D$ represent cooperation and defection, respectively.
\item \textbf{Phenotype}: An $n$-dimensional position $\mathbf{y}(k) = (y_1(k), \dots, y_n(k)) \in \mathbb{Z}^n$ in the phenotype space.
\end{itemize}
We assume that every individual can interact with every other individual in the same generation.

Now consider an interaction between individuals $k$ and $j$. If individual $k$ is a defector, then it defects regardless of the phenotypic distance to $j$. If individual $k$ is a cooperator, then it cooperates with individual $j$ with probability
\[
\varphi(k,j)=e^{-\alpha d(k,j)},
\]
where \[
d(k,j)=\|\mathbf{y}(k)-\mathbf{y}(j)\|_1=\sum_{m=1}^n |y_m(k)-y_m(j)|
\] 
is the Manhattan distance. 
This metric is appropriate because it quantifies the total number of coordinate-wise steps that separate the two phenotypes. Individual $k$ will defect, with the complementary probability $1-\varphi(k,j)$. In particular, if $d(k,j)=0$, then individual $k$ will cooperate in interaction with individual $j$ with probability $1$. If $d(k,j)>0$, cooperation will occur with probability $\varphi(k,j)$, which decreases as the phenotypic distance increases. The quantity $\varphi(k,j)$ may be interpreted as the probability that a cooperator $k$ does not recognize the phenotypic dissimilarity with individual $j$, whereas $1-\varphi(k,j)$ is the probability that it recognizes the dissimilarity and withholds cooperation. 

The parameter $\alpha > 0$ quantifies the sharpness of discrimination. Higher values of $\alpha$ indicate a highly selective environment in which cooperation is concentrated among individuals with similar phenotypes. In contrast, lower values of $\alpha$ characterize a diffuse model in which cooperation extends across a broader phenotypic range.
Figures \ref{fig:lattice} and \ref{fig:kernel_profiles} provide graphical illustrations of these concepts.

\begin{figure}
  \centering
  \includegraphics[width=\textwidth]{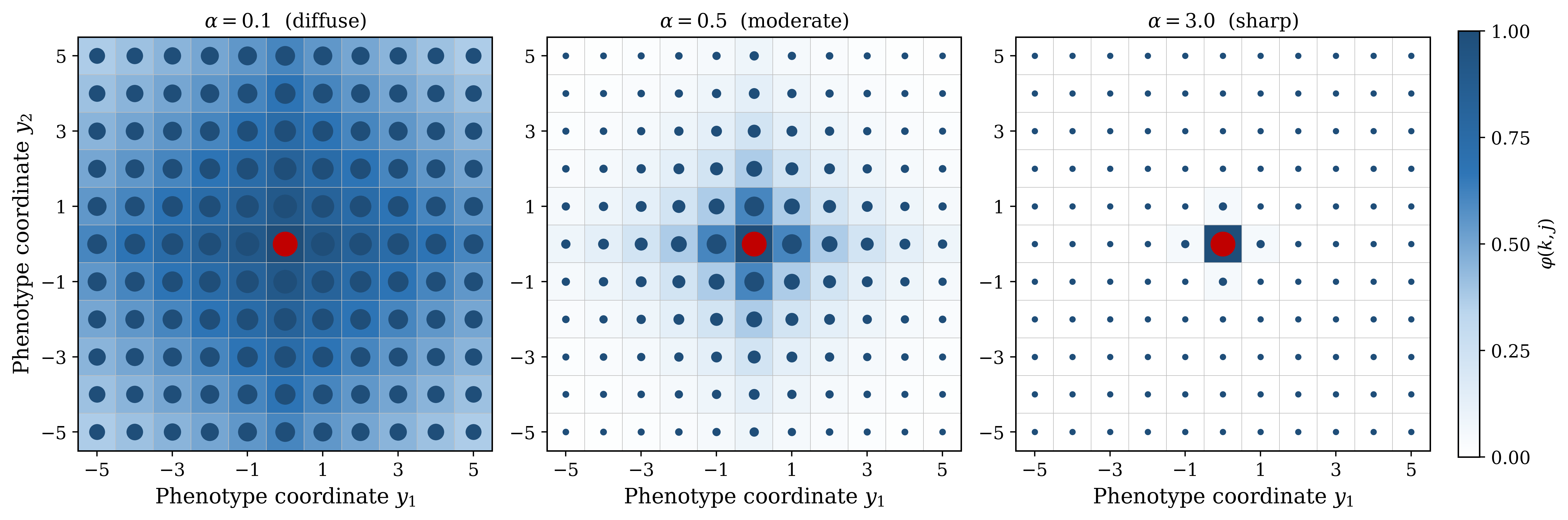}
  \caption{Illustration of the cooperation kernel on the two-dimensional phenotype lattice $\mathbb{Z}^2$. The focal cooperator (red dot) can interact with all neighboring phenotypes. The size of each node and the background shading show the probability $\varphi(k,j)$ of cooperation. As the discrimination parameter $\alpha$ increases, the cooperative tendency of individual $k$ becomes more localized with individuals of similar phenotype: for $\alpha=0.1$ (\emph{left}), cooperation is spread broadly across the lattice; for $\alpha=0.5$ (\emph{center}), it decays with phenotypic distance; and for $\alpha=3$ (\emph{right}), it is concentrated on the closest phenotypic neighbors.}
  \label{fig:lattice}
\end{figure}

\begin{figure}
  \centering
   \includegraphics[width=\textwidth,height=7cm,keepaspectratio]{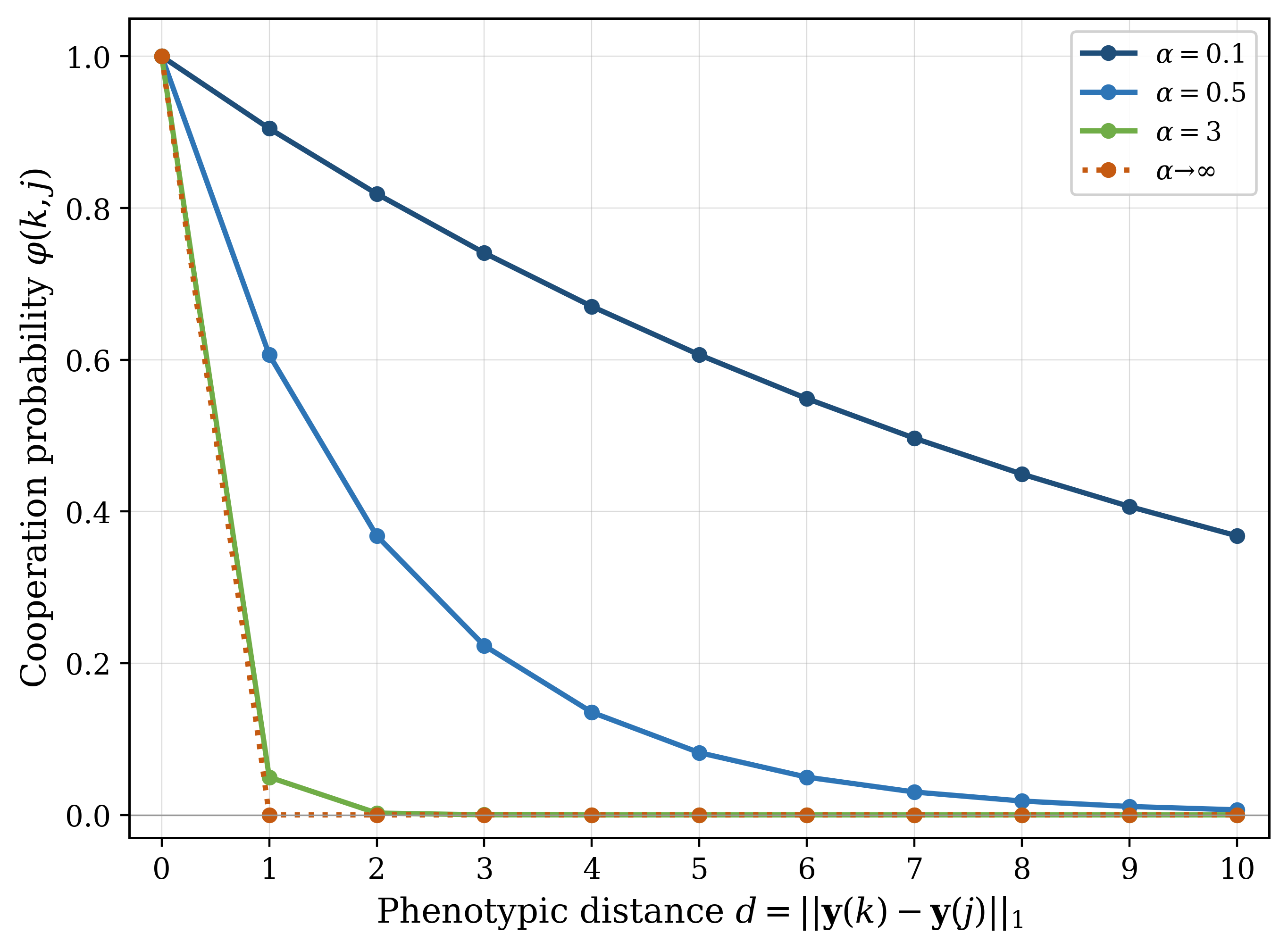}
  \caption{Curves of the cooperation probability $\varphi(k,j)$ as a function of the phenotypic distance $d(k,j)$ for $\alpha=0.1,\,0.5,\,3$ and the limiting case $\alpha\rightarrow\infty$. For small values of $\alpha$, the interaction kernel remains nearly flat, indicating low discrimination. As $\alpha$ increases, the probability mass concentrates increasingly on $d=0$, eventually converging to the binary exact-matching rule (Kronecker delta) in the limit $\alpha \to \infty$.}
  \label{fig:kernel_profiles}
\end{figure}

Suppose that interactions occur according to the simplified Prisoner's Dilemma, where a cooperator pays a cost $c>0$ to provide a benefit $b > c$ to its partner. A defector pays nothing and provides no benefit. After interactions, the expected payoff received by individual $k$ is given by 
\begin{equation}\label{eq20}
\begin{split}
a_k
&=\frac{1}{N-1}\sum_{j\neq k}\Big(\mathbf 1_{\{S(k)=S(j)=C\}}\Bigl[\varphi(k,j)\varphi(j,k)(b-c)
+(1-\varphi(k,j))\varphi(j,k)b-\varphi(k,j)(1-\varphi(j,k))c\Bigr]\\
&\quad\quad\quad\quad\quad\quad\quad\quad+\mathbf 1_{\{S(k)=D,S(j)=C\}}\varphi(k,j)b-\mathbf 1_{\{S(k)=C,S(j)=D\}}\varphi(k,j)c\Big) \\
&=\frac{1}{N-1}\sum_{j\neq k}\varphi(k,j)\bigl(b\,\mathbf 1_{\{S(j)=C\}}-c\,\mathbf 1_{\{S(k)=C\}}\bigr).
\end{split}
\end{equation}
In the last expression, the benefit to individual $k$ is discounted by the probability that its partner $j$ does not recognize their phenotypic dissimilarity, while the cost $c$ is discounted by the probability that individual $k$ does not recognize it. These payoffs determine an individual’s fertility, which is given by $f_k = 1 + \delta a_k$ for some $\delta > 0$ that measures selection intensity. We focus on weak selection, where $\delta\ll 1$.

At each generation, each individual will produce a very large number of offspring proportional to its reproductive fertility. Then, $N$ offspring will be sampled to mature and become the adults of the next generation. During reproduction, each offspring can be subject to independent mutations in the strategy and the phenotype. More precisely, an offspring will adopt the parent's strategy with probability $1-u$, or adopt a random strategy  among $\{C,D\}$ with the complementary probability $u$. In addition, the offspring will inherit the parent's phenotype with probability $1-v$. With the complementary probability $v$, a phenotypic mutation occurs. In this case, the offspring moves one step in  the positive or negative direction along a randomly chosen coordinate $m \in \{1, \dots, n\}$. In other words, the phenotype of the offspring can be
\[
\mathbf y_{\mathrm{off}}=
\begin{cases}
\mathbf y(k), & \text{with probability } 1-v,\\[4pt]
\mathbf y(k)+\mathbf e_m, & \text{with probability } \dfrac{v}{2n},\\[8pt]
\mathbf y(k)-\mathbf e_m, & \text{with probability } \dfrac{v}{2n},
\end{cases}
\]
All mutation events, both in strategy and phenotype, are independent for each offspring and across the two inherited traits. This ensures that the occurrence of a mutation in one offspring or one trait does not affect the mutation probability in another offspring or trait.

This phenotypic process generates a cluster of individuals that drifts through the lattice space. Although reproduction remains globally mixed, the distance-dependent interaction kernel ensures that phenotypic proximity leads to assortative interactions. Individuals who are close in phenotype space tend to engage in more cooperative acts than those who are far apart. Our analysis focuses on the stationary regime of this joint process to determine whether phenotypic complexity can promote or hinder the average abundance of cooperation.

\section{\label{sec2}Average abundance of cooperation}
In this section, we will determine the expression of the average abundance of cooperation in the stationary equilibrium.
Let $\mathbf{Y}_t = (\mathbf{y}_t(1), \dots, \mathbf{y}_t(N))$ be the phenotypic positions of all $N$ individuals at generation $t\geq0$. These phenotypes can shift indefinitely across the lattice, meaning that the process $(\mathbf{Y}_t)_t$ does not settle into a stationary state. However, the population tends to form a cluster that drifts collectively through the phenotype space. To describe this cluster, we will track the phenotypes of each individual $1, \dots, N-1$ relative to a reference individual $N$:
$$\mathbf{D}_t = (\mathbf{y}_t(1) - \mathbf{y}_t(N), \dots, \mathbf{y}_t(N-1) - \mathbf{y}_t(N)) \in (\mathbb{Z}^n)^{N-1}.$$
Following \citep{Kroumi2015}, the process $(\mathbf{D}_t)_t$ is irreducible and positive recurrent. When combined with strategy mutations, the joint process of relative phenotypes and strategies reaches a unique stationary distribution. We refer to this as the mutation–selection equilibrium.
In the remainder of the paper, all expectations $\mathbb E_{\delta}[.]$ are taken with respect to this equilibrium.

Now, consider a specific configuration $\mathbf{s}$ of the joint process in the equilibrium state. The frequency of cooperators is given by
$$X = \frac{1}{N} \sum_{k=1}^N \mathbf{1}_{\{S(k)=C\}}.$$ 
In the next generation, the probability $p_k(\mathbf{s})$ that an offspring descends from individual $k$ is determined by the relative fitness of $k$:
$$p_k(\mathbf{s}) = \frac{1 + \delta a_k}{\sum_{\ell=1}^N (1 + \delta a_\ell)}.$$
In this case,  the offspring is a cooperator with probability
\[(1-u)\mathbf 1_{\{S(k)=C\}}+\frac{u}{2},\]
since it inherits the parental strategy with probability $1-u$, or adopts a random strategy ($C$ or $D$ with equal probability) with probability $u$.
The expected change in cooperator frequency per generation, denoted as $\mathbb{E}_{\delta}[\Delta X \mid \mathbf{s}]$, can be written as
\begin{equation}\label{eq:DX-decomp-self}
\begin{split}
\mathbb E_{\delta}[\Delta X\mid \mathbf s]&=\sum_{k=1}^N p_k(\mathbf s)\Bigl((1-u)\mathbf 1_{\{S(k)=C\}}+\frac{u}{2}\Bigr)-X
\\
&=(1-u)\,\mathbb E_{\delta}\left[\Delta X_{\mathrm{sel}}\mid \mathbf s\right]+u\Bigl(\frac12-X\Bigr),\\
\end{split}
\end{equation}
where
\begin{equation}\label{eq1}
\mathbb E_{\delta}[\Delta X_{\mathrm{sel}}\mid \mathbf s]=\sum_{k=1}^N p_k(\mathbf s)\mathbf 1_{\{S(k)=C\}}-X
\end{equation}
denotes the expected change due to selection alone in the absence of strategy mutation.

At mutation–selection equilibrium, the expected change vanishes, $\mathbb{E}_{\delta}[\Delta X] = 0$. By taking the expectation of the decomposition in Eq. \eqref{eq:DX-decomp-self}, the average abundance of cooperation becomes
\begin{equation}\label{eq:avg-abundance-basic-self}
\mathbb{E}_{\delta}[X] = \frac{1}{2} + \frac{1 - u}{u} \mathbb{E}_{\delta}[\Delta X_{\mathrm{sel}}].
\end{equation}

In the neutral regime ($\delta=0$), we have $p_k(\mathbf s)=1/N$, for any individual $k$ and any configuration $\mathbf s$. This shows that
$$\mathbb E_0[\Delta X_{\mathrm{sel}}\mid \mathbf s]=0,$$
which shows that $\mathbb E_0[\Delta X_{\mathrm{sel}}]=0$. Then, the average abundance of cooperation is 
$$\mathbb E_0[X]=\frac12.$$
This gives the neutral baseline in the mutation–selection equilibrium

In the presence of payoff differences, we say that selection favors the average abundance of cooperation if
its average abundance exceeds its neutral value \(1/2\), which is equivalent to
\begin{equation}\label{eq:cooperation-favored-Delta}
\mathbb E_{\delta}[\Delta X_{\mathrm{sel}}]>0.
\end{equation}

To evaluate the average abundance of cooperation, note the expansion
\[
p_k(\mathbf s)
=
\frac{1+\delta a_k}{\sum_{\ell=1}^N (1+\delta a_\ell)}
=
\frac1N+\frac{\delta}{N}(a_k-\bar a)+O(\delta^2),
\]
where $\bar{a} = \frac{1}{N} \sum_{\ell=1}^N a_\ell$ is the average payoff in the population. Substituting this expansion into the expected change due to selection in Eq. \eqref{eq1}, we obtain
\begin{align}
\mathbb E_{\delta}[\Delta X_{\mathrm{sel}}\mid \mathbf s]
&=
\sum_{k=1}^N \mathbf 1_{\{S(k)=C\}}\left[\frac1N+\frac{\delta}{N}(a_k-\bar a)\right]-X+O(\delta^2)\notag\\
&=
\frac{\delta}{N}\sum_{k=1}^N \mathbf 1_{\{S(k)=C\}}(a_k-\bar a)+O(\delta^2)\notag\\
&=
\frac{\delta}{N}\sum_{k=1}^N \left(\mathbf 1_{\{S(k)=C\}}-X\right)a_k+O(\delta^2).
\label{eq:DXsel-first-order-self}
\end{align}
Inserting the payoff in \eqref{eq20} into \eqref{eq:DXsel-first-order-self} and regrouping the terms, the expected change due to selection takes the form
\begin{equation}\label{eq:DXsel-Q-self}
\mathbb E_{\delta}[\Delta X_{\mathrm{sel}}\mid \mathbf s]=\frac{\delta}{N(N-1)}\Bigl[b\,Q_{CC}^{(\alpha)}-\bigl(bX+c(1-X)\bigr)Q_C^{(\alpha)}\Bigr]+O(\delta^2),
\end{equation}
where the weighted sums are
\begin{subequations}\label{eq2}
\begin{align}
Q_C^{(\alpha)}&=\sum_{k\neq j}\varphi(k,j)\,\mathbf 1_{\{S(k)=C\}},\\
Q_{CC}^{(\alpha)}&=\sum_{k\neq j}\varphi(k,j)\,\mathbf 1_{\{S(k)=S(j)=C\}}.
\end{align}
\end{subequations}

Note that 
$$
\mathbb E_{\delta}[.]=\mathbb E_{0}[.]+\mathcal{O}(\delta).
$$
In addition, under neutrality, note the exchangeability of individuals and the symmetry between strategies $C$ and $D$. This yields
\begin{subequations}\label{eq2'}
\begin{align}
\mathbb{E}_0\left[Q_C^{(\alpha)}\right] &= \sum_{k\neq j}\mathbb{E}_0\left[\varphi(k,j)\,\mathbf 1_{\{S(k)=C\}}\right]=\frac{N(N-1)}{2} Z_n(\alpha),\\
\mathbb{E}_0\left[Q_{CC}^{(\alpha)}\right] &=\sum_{k\neq j}\mathbb{E}_0\left[\varphi(k,j)\,\mathbf 1_{\{S(k)=S(j)=C\}}\right] =\frac{N(N-1)}{2} G_n(\alpha),\\
\mathbb E_0\left[X Q_C^{(\alpha)}\right]&=\frac{1}{N}\sum_{\ell=1}^{N}\sum_{k\neq j}\varphi(k,j)\,\mathbf 1_{\{S(k)=S(\ell)=C\}}=\frac{N-1}{2}\Bigl[Z_n(\alpha)+G_n(\alpha)+(N-2)H_n(\alpha)\Bigr]
\end{align}
\end{subequations}
where
\begin{subequations}\label{eq3}
\begin{align}
Z_n(\alpha)&=\mathbb E_0\!\left[\varphi(k,j)\right], \label{eq:Zalpha-self}\\
G_n(\alpha)&=\mathbb E_0\!\left[\varphi(k,j)\,\mathbf 1_{\{S(k)=S(j)\}}\right], \label{eq:Galpha-self}\\
H_n(\alpha)&=\mathbb E_0\!\left[\varphi(k,j)\,\mathbf 1_{\{S(j)=S(\ell)\}}\right]. \label{eq:Halpha-self}
\end{align}
\end{subequations}
Here \(k\), \(j\), and \(\ell\) denote three distinct individuals chosen uniformly at random in the same generation. Taking the expectation in \eqref{eq:DXsel-Q-self} and using \eqref{eq2'}, we obtain
\begin{align}\label{eq:DeltaXsel-final-self}
\mathbb E_{\delta}[\Delta X_{\mathrm{sel}}]&=\frac{\delta}{N(N-1)}\left[b\,\mathbb E_0\left[Q_{CC}^{(\alpha)}\right]-b\,\mathbb E_0\left[XQ_C^{(\alpha)}\right]-c\left(\mathbb E_0\left[Q_C^{(\alpha)}\right]-\mathbb E_0\left[ XQ_C^{(\alpha)}\right]\right)\right]+O(\delta^2)\notag\\
&=\frac{\delta}{2N}\Big\{b\left[(N-1)G_n(\alpha)-Z_n(\alpha)-(N-2)H_n(\alpha)\right]\notag\\
&\quad\quad\quad\quad-c\left[(N-1)Z_n(\alpha)-G_n(\alpha)-(N-2)H_n(\alpha)\right]\Big\}+O(\delta^2).
\end{align}
Finally, substituting this into Eq.  \eqref{eq:avg-abundance-basic-self}, the average abundance of cooperation in the stationary equilibrium takes the form
\begin{equation}\label{eq:avg-abundance-expanded-self}
\begin{split}
\mathbb E_{\delta}[X]&=\frac12+\delta\,\frac{1-u}{2Nu}\Bigl\{b\left[(N-1)G_n(\alpha)-Z_n(\alpha)-(N-2)H_n(\alpha)\right]\\
&\quad\quad\quad\quad\quad\quad\quad\quad-c\left[(N-1)Z_n(\alpha)-G_n(\alpha)-(N-2)H_n(\alpha)\right]\Bigr\}+O(\delta^2).
\end{split}
\end{equation}

From this, we deduce that weak selection favors the abundance of cooperation if the following condition holds
\begin{equation}\label{eq:finiteN-condition-self}
b\Bigl[(N-1)G_n(\alpha)-Z_n(\alpha)-(N-2)H_n(\alpha)\Bigr] > c\Bigl[(N-1)Z_n(\alpha)-G_n(\alpha)-(N-2)H_n(\alpha)\Bigr].
\end{equation}
This condition can be restated  as $b/c > \beta_n(\alpha)$, where
\begin{equation}
\beta_n(\alpha)=\frac{(N-1)Z_n(\alpha)-G_n(\alpha)-(N-2)H_n(\alpha)}{(N-1)G_n(\alpha)-Z_n(\alpha)-(N-2)H_n(\alpha)}.
\end{equation}

To obtain an explicit expression for $\beta_n(\alpha)$ and make the analysis more tractable, we consider the limit $N \to \infty$ while the mutation probabilities $u$ and $v$ vanish ($u, v \to 0$). In this limit, we keep the mutation rates $\mu = Nu$ and $\nu = Nv$ constant. Then, the threshold becomes
\begin{equation}\label{eq10}
\beta_n(\alpha) = \frac{Z_n(\alpha) - H_n(\alpha)}{G_n(\alpha) - H_n(\alpha)}.
\end{equation}
To obtain explicit expressions for these identity measures, we will use the population's genealogy. Let $T = \tau(k,j)$ be the rescaled time back to the most recent common ancestor of individuals $k$ and $j$, taking $N$ generations as one unit of time in the limit of $N\rightarrow\infty$. Conditional on $T=\tau$, the strategy and phenotype mutations along the ancestral lines are independent. In addition, we have
\begin{subequations}
\begin{align}
 \mathbb E_0\!\left[\varphi(k,j)\mid T=\tau\right]&=e^{-2\nu\tau}[\Phi_\alpha\left(2\nu\tau/n\right)]^n,\\
 \mathbb E_0\!\left[\mathbf 1_{\{S(k)=S(j)\}}\mid T=\tau\right]&=\frac{1+e^{-2\mu\tau}}{2},
 \end{align}
\end{subequations}
where the phenotypic function $\Phi_\alpha$ is given by
\begin{equation}\label{eq0}
\Phi_\alpha(\sigma) = \frac{1 - e^{-2\alpha}}{\pi} \int_0^\pi \frac{e^{\sigma\cos\theta}}{1 - 2e^{-\alpha}\cos\theta + e^{-2\alpha}}  d\theta.
\end{equation}
Integrating over the distribution of coalescence times of three individuals $k,\, j,\,\ell$ selected randomly at the same generation, we obtain
\begin{subequations}\label{eq11}
\begin{align}
Z_n(\alpha)&=\frac{n}{2\nu}\mathcal{L}_n\!\left(\alpha,\frac{1}{2\nu}\right),\\
G_n(\alpha)&=\frac{n}{4\nu}\left[\mathcal{L}_n\!\left(\alpha,\frac{1}{2\nu}\right)+\mathcal{L}_n\!\left(\alpha,\frac{1+2\mu}{2\nu}\right)\right],\\
H_n(\alpha)&=\frac{n}{8\nu}\left[\frac{3+2\mu}{1+\mu}\,\mathcal{L}_n\!\left(\alpha,\frac{1}{2\nu}\right)+\mathcal{L}_n\!\left(\alpha,\frac{1+2\mu}{2\nu}\right)-\frac{\mu(3+2\mu)}{(1+\mu)(1+2\mu)}\,\mathcal{L}_n\!\left(\alpha,\frac{3+2\mu}{2\nu}\right)
\right].
\end{align}
\end{subequations}
Here $\mathcal{L}_n$ is a generalized Laplace transform given by
\begin{equation}\label{eq15}
\mathcal{L}_n(\alpha,x)=\int_0^\infty [\Phi_\alpha(\sigma)]^n e^{-n(1+x)\sigma}\,d\sigma.
\end{equation}
All the mathematical details regarding the expressions in \eqref{eq11} are provided in Appendix \ref{AppendixA}.
Substituting these expressions \eqref{eq11} into Eq. \eqref{eq10}, we obtain an explicit form of the threshold $\beta_n(\alpha)$ as
\begin{equation}\label{eq12}
\beta_n(\alpha)=\frac{(1+2\mu)^2\mathcal{L}_n(\alpha,\frac{1}{2\nu})+\mu(3+2\mu)\mathcal{L}_n(\alpha,\frac{3+2\mu}{2\nu})-(1+\mu)(1+2\mu)\mathcal{L}_n(\alpha,\frac{1+2\mu}{2\nu})}
{(1+\mu)(1+2\mu)\mathcal{L}_n(\alpha,\frac{1+2\mu}{2\nu})+\mu(3+2\mu)\mathcal{L}_n(\alpha,\frac{3+2\mu}{2\nu})-(1+2\mu)\mathcal{L}_n(\alpha,\frac{1}{2\nu})}.
\end{equation}

Note that the difference between the numerator and the denominator of $\beta_n(\alpha)$ in the last equation is equal to 
$$2(1+\mu)(1+2\mu)\left[\mathcal{L}_n\left(\alpha,\frac{1}{2\nu}\right)-\mathcal{L}_n\left(\alpha,\frac{1+2\mu}{2\nu}\right)\right]>0.$$ 
This follows from the fact that $x\rightarrow\mathcal{L}_n(\alpha,x)$ is decreasing. Consequently, we have $\beta_n(\alpha)>1$.


\section{\label{sec3}Effect of the discrimination parameter $\alpha$}
In this section, we examine how the sharpness of phenotypic discrimination, measured by parameter $\alpha$, affects the threshold $\beta_n(\alpha)$. 

Our first result is the following (the proof is given in Appendix \ref{appendixB}).
\begin{Result}\label{thm:alpha-monotonicity-section}
Fix a dimension $n \ge 1$, mutation rates $\mu > 0$, $\nu > 0$. An increase in the parameter $\alpha$ will decrease the threshold $\beta_n(\alpha)$ required for selection to favor the abundance of cooperation.
\end{Result}

This result shows that the sharpest discrimination (high values of $\alpha$) is the most favorable for the abundance of cooperation. At first glance, one might expect broader altruism to expand cooperation. Instead, the opposite holds: cooperation is costly, and helping phenotypically distant individuals is less aligned with phenotypic 
assortment. This yields a lower selective benefit than helping highly similar partners. As $\alpha$ increases, cooperation concentrates among similar phenotypes, and this concentration effect outweighs the 
benefit of having more partners.

Now, we will consider the two limiting regimes $\alpha \to 0^+$ and $\alpha \to \infty$. 
\begin{itemize}
\item \textbf{The Zero-Discrimination Limit ($\alpha \to 0^+$).} As $\alpha$ approaches $0$, we have $\varphi(k,j) \to 1$ for all $k \neq j$, which yields $Z_n(0) = 1$. This indicates that the behavior of any individual is independent of its partner's phenotype, and that no phenotypic assortment remains. 
It follows that selection does not favor the abundance of cooperation \cite{Antal2009a}, showing that $\beta_n(0^+)=\infty$. 

In the same limit, we have  $$G_n(0)\approx H_n(0) \approx \frac{1+\mu}{1+2\mu}.$$ 
This is a consequence of the fact that 
\begin{equation}\label{eq31}
\mathcal{L}_n(0,x) \sim \int_0^\infty e^{-nx\sigma} \, d\sigma = \frac{1}{nx}.
\end{equation}
This is exactly the probability that two randomly chosen individuals share the same behavioral strategy in a large, well-mixed population under neutrality (see \cite[Eq. (26a)]{KroumiLessard2021} or \cite[Eq. (18)]{KroumiLessard2022JMB}).
\item \textbf{The Perfect-Discrimination Limit ($\alpha \to \infty$).} As $\alpha$ becomes very large, we have $e^{-\alpha d} \to \mathbf{1}_{\{d=0\}}$, so that cooperation is extended only to individuals with identical phenotypes. In this limit, the Laplace transform $\mathcal{L}_n$ converges to the transform given in \citep[Eq. (23)]{Kroumi2015} for the exact-matching model. Thus, we have
$$\lim_{\alpha\rightarrow\infty}\beta_n(\alpha) = \beta_n,$$
where $\beta_n$ is the threshold derived in \citep{Kroumi2015}. This is the lowest value of the threshold $\beta_n(\alpha)$ for any $\alpha>0$.
By Result \ref{thm:alpha-monotonicity-section}, the exact-matching model is the limiting minimizer of the threshold $\beta_n(\alpha)$ within this exponential family. 
\end{itemize}
These limiting cases have a clear evolutionary interpretation. When $\alpha$ is small, individuals extend cooperative acts broadly regardless of phenotypic distance, which dilutes assortment and prevents phenotypic similarity from functioning. As $\alpha$ increases, benefits become concentrated among phenotypically similar individuals, and the resulting assortment is more efficient.
These conclusions are illustrated in Figure \ref{fig:monotone_alpha}.
\begin{figure}
    \centering
    \includegraphics[width=\textwidth,height=7cm,keepaspectratio]{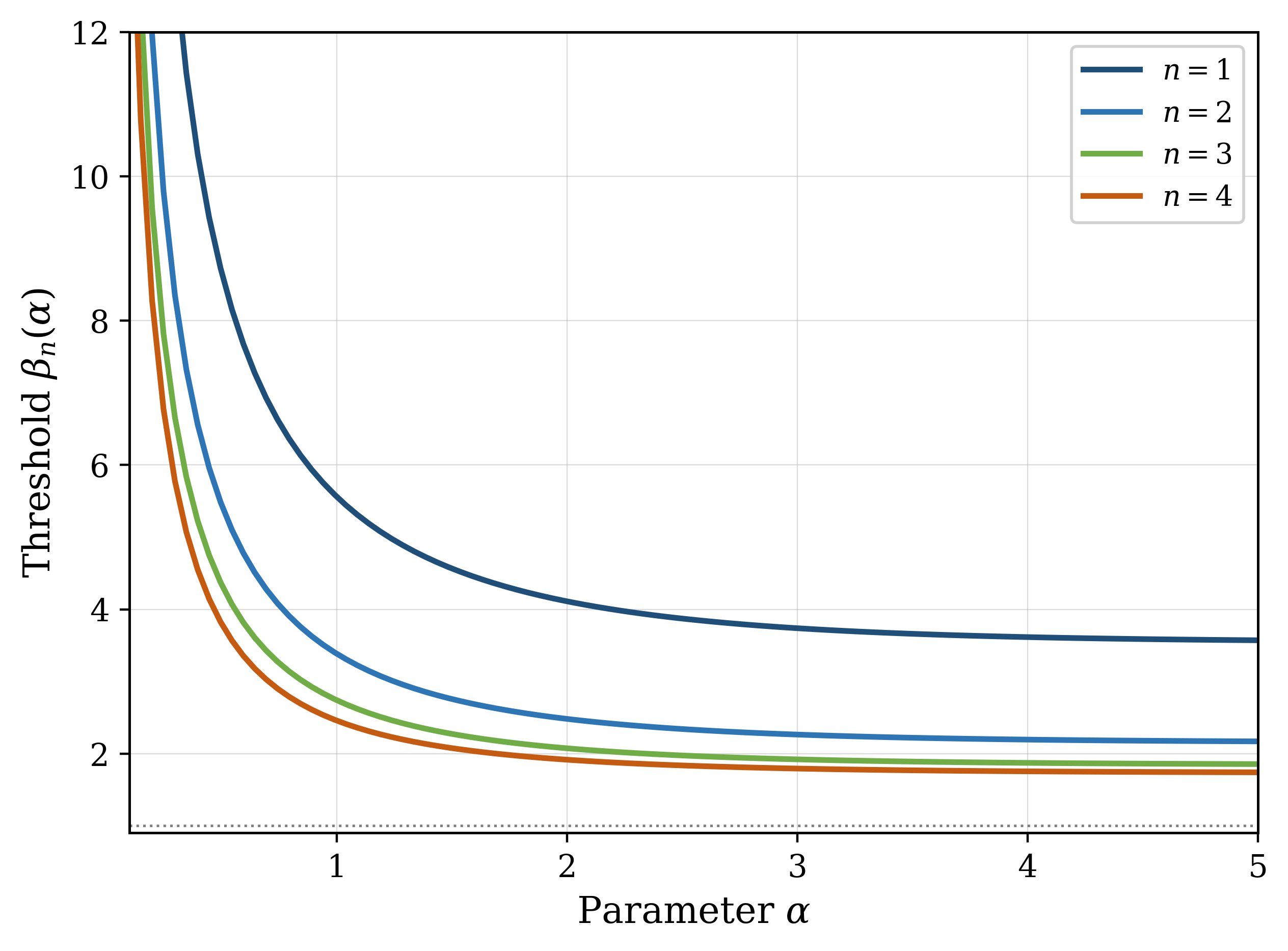}
    \caption{The threshold $\beta_n(\alpha)$ as a function of the parameter $\alpha$. The curves are shown for dimensions $n=1,2,3,4$, with mutation parameters fixed at $\mu=1$ and $\nu=3$. In each case, the threshold decreases as $\alpha$ increases.}
    \label{fig:monotone_alpha}
\end{figure}


\section{\label{sec4}Effect of phenotype-space dimension $n$}
In this section, we will study the effect of the dimensionality of the phenotype space on the threshold $\beta_n(\alpha)$. Increasing $n$ alters the geometry of the phenotype space. As a result, phenotypic mutations disperse across more independent directions. This makes the set of phenotypes in any given individual's neighborhood relatively sparse. 

This point is addressed in the following result (see Appendix \ref{AppendixC} for the proof).
\begin{Result}\label{thm:n-monotonicity-section}
For all $\alpha > 0$, $\mu > 0$, and $\nu > 0$, the critical benefit-to-cost ratio satisfies
$\beta_{n+1}(\alpha) < \beta_n(\alpha)$,
for every integer $n \ge 1$. A one-dimensional phenotype space is the least favorable configuration, and each additional phenotypic coordinate strictly lowers the threshold $\beta_n(\alpha)$.
\end{Result}

 In higher-dimensional phenotype spaces, a single mutation displaces an individual along a single coordinate, leaving all other coordinates unchanged. Lineages accumulating mutations in different directions diverge along independent axes and quickly move out of each other's interaction range. Since the kernel $\varphi$ decays with phenotypic distance, this spread weakens social ties between diverging lineages. At the same time, it strengthens the assortment by phenotype among individuals who remain close.

Consider the two scenarios defining the bounds of phenotypic complexity.
\begin{itemize}
\item \textbf{Case $n=1$}. This scenario represents a one-dimensional trait axis where mutations accumulate along a single line. This results in maximal overlap between lineages and the weakest phenotypic filtering of social interactions. 
\item \textbf{Case $n=\infty$}. In this limit, phenotypic mutations disperse over a large number of coordinate directions. Independently mutating lineages separate along different axes. As a result, local phenotypic neighborhoods become increasingly sparse. This geometric dilution strengthens phenotype-based assortment. The limiting threshold is
\begin{equation}\label{eq:beta-infty-alpha-section}
\beta_{\infty}(\alpha) = \lim_{n\to\infty} \beta_n(\alpha) = \frac{4c_\alpha^2\nu^2 + (8\mu+6)c_\alpha\nu + 4\mu^2 + 8\mu + 3}{4c_\alpha\nu(c_\alpha\nu + \mu + 1)}>1.
\end{equation}
where $c_{\alpha} = 1 - e^{-\alpha} \in (0,1)$ (see Appendix \ref{AppendixD-part3} for the proof).
This limit represents the minimum threshold achievable for any fixed set of parameters $\alpha,\,\mu,\,\nu$ for all finite dimensions $n\geq 1$.

Note that \eqref{eq:beta-infty-alpha-section} demonstrates that the large-$n$ limit remains sensitive to $\alpha$. Thus, increasing dimensionality does not make discrimination sharpness irrelevant. In addition, we have $\partial \beta_{\infty}(\alpha)/\partial\alpha<0$. So, an increase in $\alpha$ lowers the threshold and makes it more likely for selection to favor the abundance of cooperation. If discrimination is high ($\alpha \to \infty$), $c_\alpha \to 1$ and the threshold reaches its global minimum, given by
\begin{equation}
\beta_{\infty}(\infty) = \frac{4\nu^2 + (8\mu+6)\nu + 4\mu^2 + 8\mu + 3}{4\nu(\nu + \mu + 1)}.
\end{equation}
However, as the ability to discriminate vanishes ($\alpha \to 0$), $c_\alpha$ approaches $0$ and the threshold diverges to $\infty$. Dimension and discrimination therefore act on different aspects of the same geometry: increasing $n$ thins phenotypic neighborhoods, while increasing $\alpha$ concentrates cooperative acts within them.
\end{itemize}
These conclusions are illustrated in Figure \ref{fig:monotone_n}.

\begin{figure}
    \centering
     \includegraphics[width=\textwidth,height=7cm,keepaspectratio]{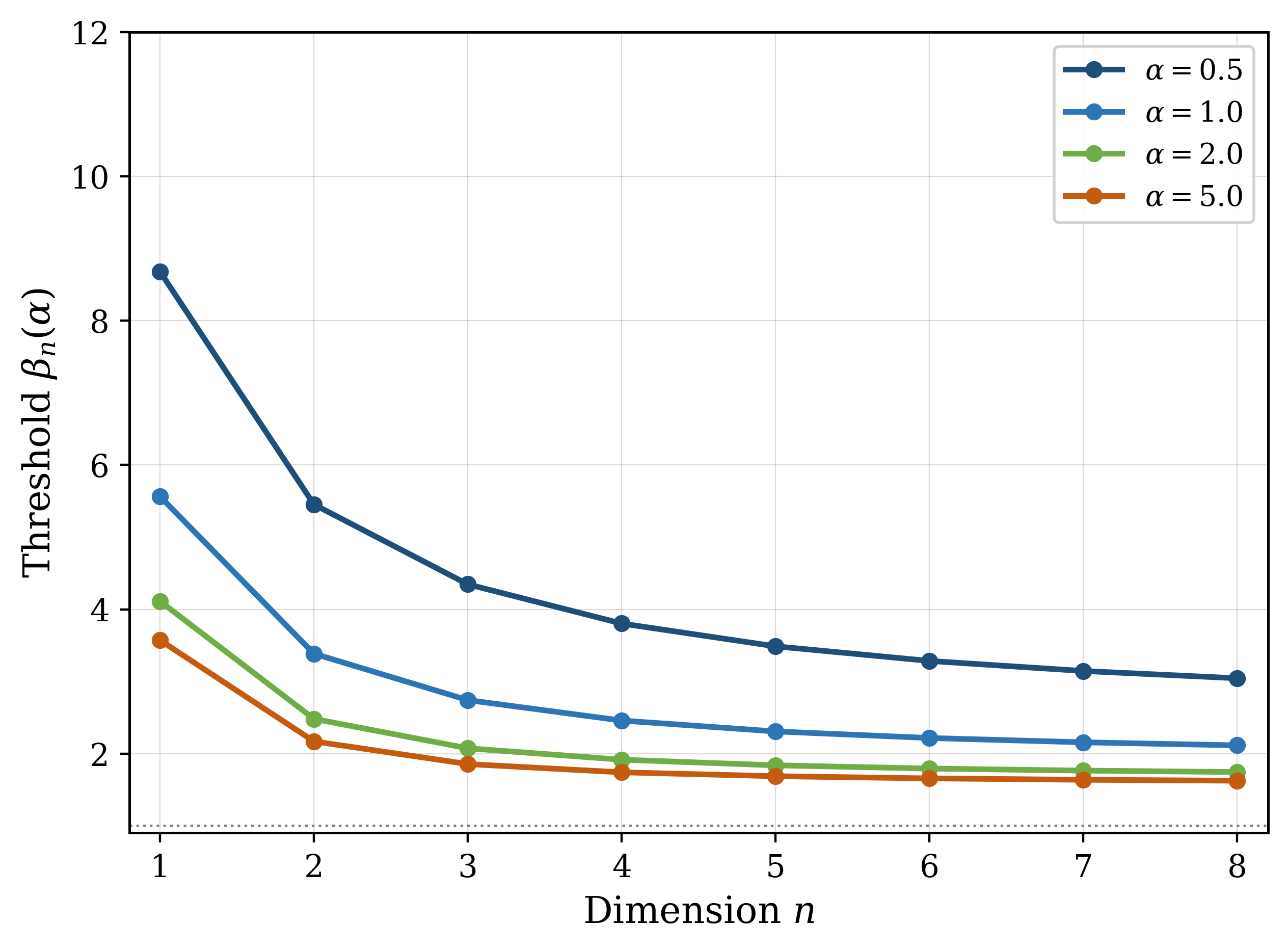}
    \caption{Threshold $\beta_n(\alpha)$ as a function of the dimension $n \in \{1, \dots, 8\}$ for various levels of parameter $\alpha \in \{0.5, 1, 2, 5\}$, with $\mu = 1$ and $\nu = 3$. Each curve is strictly decreasing, confirming Result \ref{thm:n-monotonicity-section}.}
    \label{fig:monotone_n}
\end{figure}

\section{\label{sec5}Asymptotic regimes in the mutation parameters}
Now, we analyze the threshold $\beta_n(\alpha)$ in four asymptotic regimes of the mutation parameters $\mu$ and $\nu$.
\subsection{Regime 1: $\nu\rightarrow 0$}
Suppose that the phenotype mutation rate is rare ($\nu \to 0$). The generalized Laplace transform has the asymptotic expansion
\begin{equation}\label{eq:Gtilde-low-nu-section}
\mathcal{L}_n\!\left(\alpha,\frac{a}{2\nu}\right)=\frac{2\nu}{na}-\frac{4c_\alpha\nu^2}{na^2}+O(\nu^3),
\end{equation}
for any $a > 0$ (see Appendix \ref{AppendixD-part1}). Evaluating this expansion at $a=1,\,1+2\mu,\,3+2\mu$ and substituting the results in Eq. \eqref{eq12} yields
\begin{align}
\beta_n(\alpha)\approx\frac{(1+2\mu)(3+2\mu)}{4(1+\mu)c_\alpha}\,\frac{1}{\nu}.\label{eq:lownu}
\end{align}
This shows that a small phenotype mutation rate will make it less likely for selection to favor the abundance of cooperation. When phenotypic mutations are rare, the population remains clustered within a small number of phenotype classes for extended periods, so cooperators and defectors repeatedly interact within the same groups. As a result, phenotypic assortment is too weak to protect cooperators from exploitation.

The following result summarizes all the findings.
\begin{Result}
When phenotype mutation is rare ($\nu \ll 1$), the threshold required for selection to favor the abundance of cooperation diverges as
\begin{equation*}
\beta_n(\alpha) \approx \frac{(1+2\mu)(3+2\mu)}{4(1+\mu)c_\alpha} \frac{1}{\nu}.
\end{equation*}
The factor $c_\alpha$ captures the additional selective cost imposed by weak discrimination.
\end{Result}

\subsection{Regime 2: $\nu\rightarrow\infty$}
Consider the case of high phenotype mutation rate \(\nu\to\infty\).
In this limit, we have the asymptotic behavior
\begin{equation}\label{eq14}
\lim_{\nu\to\infty}\beta_n(\alpha)=1,
\end{equation}
for $n\geq2$. See Appendix \ref{AppendixD-part2} for the mathematical details. This shows that, in this asymptotic regime, selection will favor the abundance of cooperation as long as $b>c$, and
that neither the strategy mutation rate $\mu$ nor the parameter $\alpha$ impacts this condition.

In the one-dimensional case ($n=1$), we obtain
\begin{equation}\label{eq:beta-high-nu-n1-section}
\lim_{\nu\to\infty}\beta_1(\alpha) = \frac{(1+2\mu)^2 + \mu\sqrt{3+2\mu} - (1+\mu)\sqrt{1+2\mu}}{\mu\sqrt{3+2\mu} + (1+\mu)\sqrt{1+2\mu} - (1+2\mu)}>1.
\end{equation}
Contrary to the case of $n\geq2$, the rate $\mu$ impacts this threshold. 
A simple differentiation confirms that this limiting threshold increases with the strategy mutation rate $\mu$, with
\begin{subequations}\label{eq:beta-high-nu-n1-limits}
\begin{align}
&\lim_{\mu\rightarrow\infty}\lim_{\nu\to\infty}\beta_1(\alpha) = \lim_{\mu\rightarrow\infty} \frac{ 4\mu^2 + O(\mu^{3/2})}{ 2\sqrt{2}\mu^{3/2} + O(\mu)}=\infty,\\
&\lim_{\mu\rightarrow0^+}\lim_{\nu\to\infty}\beta_1(\alpha) =\lim_{\mu\rightarrow0^+} \frac{ (2 + \sqrt{3})\mu + O(\mu^2)}{ \sqrt{3}\mu + O(\mu^2)}=1+\frac{2}{ \sqrt{3}}.
\end{align}
\end{subequations}
This shows that even with a high phenotype mutation rate, a high strategy mutation rate will make it impossible for selection to favor the abundance of cooperation.
The most favorable conditions arise when a high phenotype mutation is paired with a low strategy mutation. This will prevent defectors from disrupting cooperative clusters.

We consolidate these findings in the following result
\begin{Result}
In the regime of frequent phenotype mutation ($\nu \gg 1$), the threshold satisfies $\beta_n(\alpha) \approx 1$ for all $n \geq 2$, whereas for $n=1$ it converges 
to a finite limit strictly greater than $1$ given in \eqref{eq:beta-high-nu-n1-section}. Dimension $n=2$ is therefore the critical value 
above which rapid phenotypic turnover drives the threshold to its lower bound $1$.
\end{Result}

\subsection{Regime 3: $\mu\rightarrow0$}
Focus now on the case of rare strategy mutation $\mu \to 0^+$. 
Using Taylor expansion, we have
\begin{subequations}
\begin{align}
\mathcal{L}_n\!\left(\alpha, \frac{1+2\mu}{2\nu}\right) &= \mathcal{L}_n\!\left(\alpha, \frac{1}{2\nu}\right) + \frac{\mu}{\nu} \frac{\partial}{\partial x}\mathcal{L}_n\left(\alpha, \frac{1}{2\nu}\right) + O(\mu^2),\\
\mathcal{L}_n\!\left(\alpha, \frac{3+2\mu}{2\nu}\right) &= \mathcal{L}_n\!\left(\alpha, \frac{3}{2\nu}\right) + O(\mu).
\end{align}
\end{subequations}
where 
\begin{equation}
\frac{\partial}{\partial x}\mathcal{L}_n(\alpha, x) = -n \int_0^\infty \sigma \, [\Phi_\alpha(\sigma)]^n e^{-n(1+x)\sigma}\,d\sigma<0.
\end{equation}
Substituting these expansions into \eqref{eq12}, we obtain
\begin{equation}\label{eq:beta-low-mu-section-revised}
\beta_n(\alpha)=\frac{\mathcal{L}_n\!\left(\alpha,\frac{1}{2\nu}\right)-\frac{1}{\nu}\frac{\partial}{\partial x}\mathcal{L}_n\left(\alpha, \frac{1}{2\nu}\right)+3\,\mathcal{L}_n\!\left(\alpha,\frac{3}{2\nu}\right)}
{\mathcal{L}_n\!\left(\alpha,\frac{1}{2\nu}\right)+\frac{1}{\nu}\frac{\partial}{\partial x}\mathcal{L}_n\left(\alpha, \frac{1}{2\nu}\right)+3\,\mathcal{L}_n\!\left(\alpha,\frac{3}{2\nu}\right)},
\qquad\text{as}\quad\mu\rightarrow0^+.
\end{equation}
This threshold is finite, making it possible for selection to favor the abundance of cooperation. If strategy mutation is rare, cooperative clusters persist without frequent disruption by defectors, and the threshold is determined only by the phenotype interaction structure. Note that $\lim_{\mu\rightarrow0^+}\beta_n(\alpha)>1$ confirms that low strategy mutation reduces but does not remove completely the barrier to cooperation. A positive degree of phenotypic assortment remains necessary.

These findings are summarized in the following result.
\begin{Result}
When strategy mutation is rare ($\mu \ll 1$), the threshold $\beta_n(\alpha)$ converges to the finite limit given by \eqref{eq:beta-low-mu-section-revised}.
\end{Result}

\subsection{Regime 4: $\mu\rightarrow\infty$}
Suppose that strategy mutation is highly frequent ($\mu \to \infty$). In this limit, using similar expansion as \eqref{eq:Gtilde-low-nu-section}, we get
\begin{align*}
\mathcal{L}_n\left(\alpha,\frac{1+2\mu}{2\nu}\right)&=\frac{2\nu}{n(1+2\mu)}+O(\mu^{-2}),\\
\mathcal{L}_n\left(\alpha,\frac{3+2\mu}{2\nu}\right)&=\frac{2\nu}{n(3+2\mu)}+O(\mu^{-2}).
\end{align*}
Substituting the results into \eqref{eq12}, we obtain
\begin{equation}
\beta_n(\alpha)=\frac{2\mathcal{L}_n\!\left(\alpha,\frac{1}{2\nu}\right)}{\frac{2\nu}{n}-\mathcal{L}_n\!\left(\alpha,\frac{1}{2\nu}\right)}\,\mu+O(1).
\end{equation}
Since
\[
\mathcal{L}_n\!\left(\alpha,\frac{1}{2\nu}\right)<\int_0^\infty e^{-n\sigma/(2\nu)}\,d\sigma=\frac{2\nu}{n},
\]
the coefficient of $\mu$ is strictly positive, and accordingly
\begin{equation}\label{eq:beta-high-mu-section}
\lim_{\mu\rightarrow\infty}\beta_n(\alpha)=\infty.
\end{equation}
This shows that high strategy mutation continuously reintroduces defectors into phenotype classes that would otherwise remain dominated by cooperators. This breaks down the correlation between phenotype and behavioral type that the assortment mechanism depends on. 

These observations are summarized in the following.
\begin{Result} 
In the case of high strategy mutation rates, the threshold $\beta_n(\alpha)$ diverges linearly with $\mu$, making it impossible for selection to favor the abundance of cooperation.
\end{Result}

\subsection{Summary}
These four regimes together describe how the threshold $\beta_n(\alpha)$ depends on the mutation parameters. If the phenotypic mutation is rare $(\nu \to 0^+)$, the threshold diverges like $1/\nu$, so cooperation cannot be favored. However, when phenotypic mutation is frequent $(\nu \to \infty)$, the threshold approaches $1$ for all $n \ge 2$, so selection favors the abundance of cooperation whenever $b>c$, independently of $\alpha$ and $\mu$. The strategy mutation rate has the opposite effect. If $\mu$ is small, the threshold remains finite and depends on the phenotype-based interaction structure, so cooperation can be favored if $b/c$ is sufficiently large. When $\mu$ is large, the threshold grows linearly with $\mu$. In this regime, cooperation becomes impossible to favor. 
These behaviors are shown in Figure \ref{fig:phase_nu_mu}.
\begin{figure}[H]
  \centering
 \includegraphics[width=\textwidth,height=8cm,keepaspectratio]{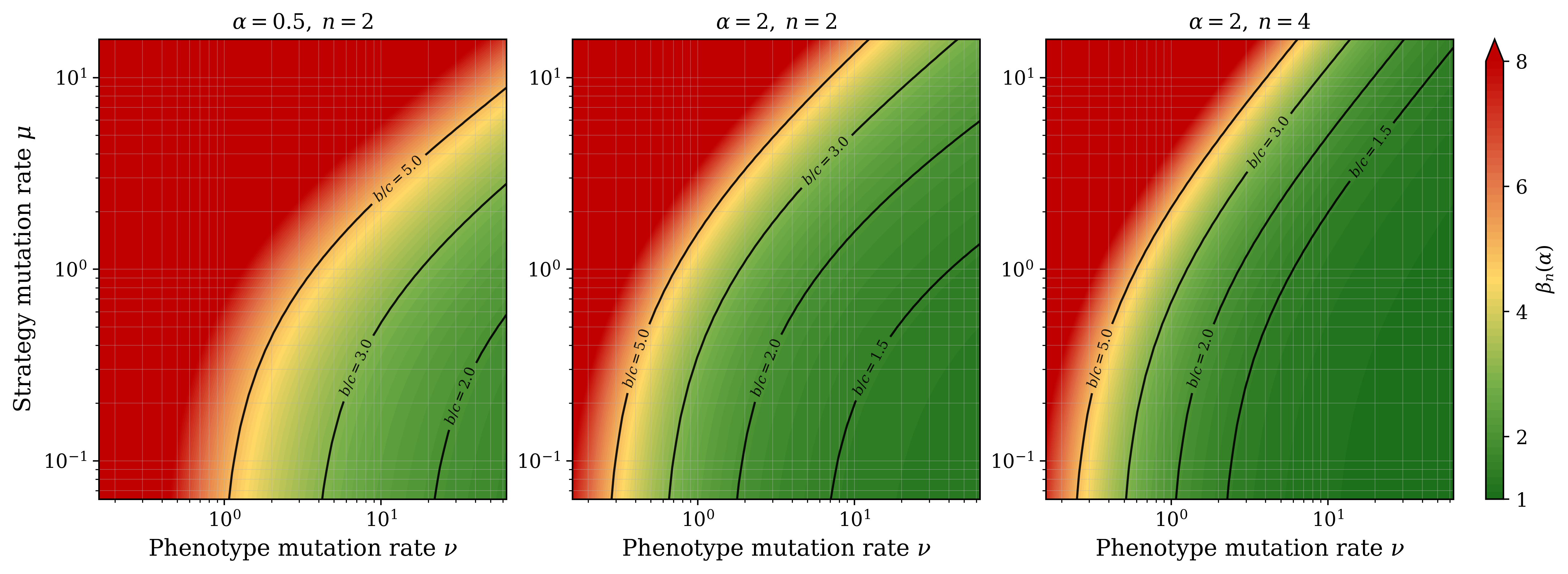}
  \caption{Phase diagrams of the cooperation threshold $\beta_n(\alpha)$ in the $(\nu,\mu)$ plane (logarithmic scales). Across all panels, low phenotypic mutation $\nu$ is unfavorable, with the threshold diverging as $1/\nu$, whereas a large $\nu$ promotes cooperation and, for $n\ge 2$, drives the threshold toward the lower bound $1$. However, a low strategy mutation rate $\mu$ yields a finite threshold, while a large $\mu$ suppresses cooperation by causing the threshold to grow linearly. \emph{Left} ($\alpha=0.5, n=2$): soft discrimination restricts the cooperative region. \emph{Center} ($\alpha=2, n=2$): sharper discrimination expands it. \emph{Right} ($\alpha=2, n=4$): higher phenotypic dimensionality expands it further.}
  \label{fig:phase_nu_mu}
\end{figure}

\section{\label{sec6}General symmetric two-player games}
The analytical framework developed above for the simplified Prisoner's dilemma can be extended to a general symmetric two-player game 
\[
\begin{array}{c|cc}
 & C & D\\
\hline
C & R & S\\
D & T & P
\end{array}.
\]

After interactions, the expected payoff obtained by individual \(k\) in interaction with individual \(j\) becomes
\[
R\,q_{kj}q_{jk}+S\,q_{kj}(1-q_{jk})+T\,(1-q_{kj})q_{jk}+P\,(1-q_{kj})(1-q_{jk}),
\]
where \(q_{\ell_1\ell_2}=\mathbf 1_{\{S(\ell_1)=C\}}\varphi(\ell_1,\ell_2)\) denotes the probability that individual $\ell_1$ cooperates with individual $\ell_2$. Expanding it and using \((\mathbf 1_{\{S(k)=C\}})^2=\mathbf 1_{\{S(k)=C\}}\) and \(\varphi(j,k)=\varphi(k,j)\), we obtain
\begin{align}
a_k&=\frac{1}{N-1}\sum_{j\neq k}\Bigl[R\,\mathbf 1_{\{S(k)=S(j)=C\}}\varphi(k,j)^2+S\,\mathbf 1_{\{S(k)=C\}}\varphi(k,j)(1-\mathbf 1_{\{S(j)=C\}}\varphi(k,j))\notag\\
&\quad+T\,(1-\mathbf 1_{\{S(k)=C\}}\varphi(k,j))\mathbf 1_{\{S(j)=C\}}\varphi(k,j)+P\,(1-\mathbf 1_{\{S(k)=C\}}\varphi(k,j))(1-\mathbf 1_{\{S(j)=C\}}\varphi(k,j))\Bigr]\notag\\
&=\frac{1}{N-1}\sum_{j\neq k}\Bigl[P+(S-P)\mathbf 1_{\{S(k)=C\}}\varphi(k,j)+(T-P)\mathbf 1_{\{S(j)=C\}}\varphi(k,j)\notag\\
&\quad\quad\quad\quad\quad\quad\quad\quad\quad\quad\quad\quad\quad\quad\quad+(R-S-T+P)\mathbf 1_{\{S(k)=S(j)=C\}}\varphi(k,j)^2\Bigr].
\label{eq:general-payoff-prob}
\end{align}
The simplified Prisoner's Dilemma is recovered by taking \(R=b-c\), \(S=-c\), \(T=b\), and \(P=0\), in which case \(R-S-T+P=0\) and the quadratic term disappears. However, for a general game, the nonlinear term proportional to \(\varphi(k,j)^2\) is essential and reflects the fact that both players must independently cooperate in order for the mutual-cooperation payoff \(R\) to be realized.

Under weak selection, the conditional change in the frequency of cooperation due to selection becomes
\begin{equation}\label{eq:sumcka-general-prob}
\begin{split}
&\mathbb E_{\delta}[\Delta X_{\mathrm{sel}}\mid \mathbf s]\\
&=\frac{\delta}{N}\sum_{k=1}^N (\mathbf 1_{\{S(k)=C\}}- X)\,a_k+O(\delta^2)\\
&=\frac{\delta}{N}\Bigg[\frac{S-P}{N-1}\sum_{k\neq j}\varphi(k,j)\mathbf 1_{\{S(k)=C\}}(\mathbf 1_{\{S(k)=C\}}-X)+\frac{T-P}{N-1}\sum_{k\neq j}\varphi(k,j)\mathbf 1_{\{S(j)=C\}}(\mathbf 1_{\{S(k)=C\}}-X)\\
&\qquad+\frac{R-S-T+P}{N-1}\sum_{k\neq j}\varphi(k,j)^2\mathbf 1_{\{S(k)=S(j)=C\}}(\mathbf 1_{\{S(k)=C\}}- X)\Bigg]\\
&=\frac{\delta}{N(N-1)}\Bigl[(S-P)(1- X)Q_C^{(\alpha)}+(T-P)\bigl(Q_{CC}^{(\alpha)}- X\,Q_C^{(\alpha)}\bigr)+(R-S-T+P)(1-X)Q_{CC}^{(2\alpha)}
\Bigr],
\end{split}
\end{equation}
where $Q_C^{(\alpha)}$ and $Q_{CC}^{(\alpha)}$ are given in Eq. \eqref{eq2} and
\[
Q_{CC}^{(2\alpha)}=\sum_{k\neq j}\varphi(k,j)^2\mathbf 1_{\{S(k)=S(j)=C\}}.
\]
Using Eq. \eqref{eq2'}, we have
\begin{subequations}\label{eq30}
\begin{align}
\mathbb E_0[(1- X)Q_C^{(\alpha)}]&=\frac{N-1}{2}\Bigl[(N-1)Z_n(\alpha)-G_n(\alpha)-(N-2)H_n(\alpha)\Bigr],\label{eq:EC-general-prob}\\
\mathbb E_0[Q_{CC}^{(\alpha)}- X\,Q_C^{(\alpha)}]&=\frac{N-1}{2}\Bigl[(N-1)G_n(\alpha)-Z_n(\alpha)-(N-2)H_n(\alpha)\Bigr].\label{eq:ECC-general-prob}
\end{align}
\end{subequations}
The term $XQ_{CC}^{(2\alpha)}$ has no counterpart in the Prisoner's Dilemma case and requires introducing a three-point identity
\[
\eta_n(\gamma)=\mathbb E_0\!\left[e^{-\gamma\|\mathbf y(k)-\mathbf y(j)\|_1}\,\mathbf 1_{\{S(k)=S(j)=S(\ell)\}}\right],
\]
for $\gamma>0$. Using \(\varphi(k,j)^2=e^{-2\alpha\|\mathbf y(k)-\mathbf y(j)\|_1}\), we obtain
\begin{equation}\label{eq:E1mXQCC2-general-prob}
\mathbb E_0\left[(1- X)Q_{CC}^{(2\alpha)}\right]=\frac{(N-1)(N-2)}{2}\Bigl[G_n(2\alpha)-\eta_n(2\alpha)\Bigr].
\end{equation}
Taking the expectation of \eqref{eq:sumcka-general-prob} and using \eqref{eq30}-\eqref{eq:E1mXQCC2-general-prob}, we have
\begin{equation}\label{eq:general-DeltaXsel-prob}
\begin{split}
\mathbb E_{\delta}[\Delta X_{\mathrm{sel}}]&=\frac{\delta}{2N}\Bigl\{(S-P)\Bigl[(N-1)Z_n(\alpha)-G_n(\alpha)-(N-2)H_n(\alpha)\Bigr]\\
&\qquad\qquad+(T-P)\Bigl[(N-1)G_n(\alpha)-Z_n(\alpha)-(N-2)H_n(\alpha)\Bigr]\\
&\qquad\qquad+(R-S-T+P)(N-2)\Bigl[G_n(2\alpha)-\eta_n(2\alpha)\Bigr]\Bigr\}+O(\delta^2).
\end{split}
\end{equation}
Therefore, the average abundance of cooperation in the stationary equilibrium will take the form
\begin{equation}\label{eq:general-abundance-prob}
\begin{split}
\mathbb E_{\delta}[X]&=\frac12+\delta\,\frac{1-u}{2Nu}\Bigl\{(S-P)\Bigl[(N-1)Z_n(\alpha)-G_n(\alpha)-(N-2)H_n(\alpha)\Bigr]\\
&\qquad\qquad\qquad\qquad+(T-P)\Bigl[(N-1)G_n(\alpha)-Z_n(\alpha)-(N-2)H_n(\alpha)\Bigr]\\
&\qquad\qquad\qquad\qquad+(R-S-T+P)(N-2)\Bigl[G_n(2\alpha)-\eta_n(2\alpha)\Bigr]\Bigr\}+O(\delta^2).
\end{split}
\end{equation}

In the large-population limit \(N\to\infty\), \(u\to0\), \(v\to0\), with \(\mu=Nu\) and \(\nu=Nv\) kept fixed, we have
\begin{subequations}
\begin{align}
Z_n(\gamma)&=\frac{n}{2\nu}\,\mathcal L_n\!\left(\gamma,\frac{1}{2\nu}\right),\\
G_n(\gamma)&=\frac{n}{4\nu}\left[\mathcal L_n\!\left(\gamma,\frac{1}{2\nu}\right)+\mathcal L_n\!\left(\gamma,\frac{1+2\mu}{2\nu}\right)\right],\\
H_n(\gamma)&=\frac{n}{8\nu}\left[\frac{3+2\mu}{1+\mu}\,\mathcal L_n\!\left(\gamma,\frac{1}{2\nu}\right)+\mathcal L_n\!\left(\gamma,\frac{1+2\mu}{2\nu}\right)-\frac{\mu(3+2\mu)}{(1+\mu)(1+2\mu)}\,\mathcal L_n\!\left(\gamma,\frac{3+2\mu}{2\nu}\right)\right],\\
\eta_n(\gamma)&=\frac{n}{8\nu}\left[\frac{2+\mu}{1+\mu}\,\mathcal L_n\!\left(\gamma,\frac{1}{2\nu}\right)+2\mathcal L_n\!\left(\gamma,\frac{1+2\mu}{2\nu}\right)-\frac{\mu(3+2\mu)}{(1+\mu)(1+2\mu)}\,\mathcal L_n\!\left(\gamma,\frac{3+2\mu}{2\nu}\right)\right].
\end{align}
\end{subequations}
For details regarding the calculation of the identity $\eta_n(\alpha)$, see Appendix \ref{AppendixA}.
We find that weak selection favors the abundance of cooperation if the following condition holds
\begin{equation}\label{eq:general-largeN-condition-prob}
(S-P)\bigl[Z_n(\alpha)-H_n(\alpha)\bigr]+(T-P)\bigl[G_n(\alpha)-H_n(\alpha)\bigr]+(R-S-T+P)\bigl[G_n(2\alpha)-\eta_n(2\alpha)\bigr]>0.
\end{equation}


\section{\label{sec8}Discussion}
In this paper, we have considered a finite population of size \(N\) in which cooperation operates through a graded recognition rule. Each individual is characterized by an \(n\)-dimensional phenotype and a strategy, cooperation or defection. A cooperator helps another individual in the population with probability \(e^{-\alpha d}\), where \(d\) is the distance between their phenotypes and \(\alpha > 0\) measures discrimination sharpness. In doing so, it pays a cost \(c>0\) to provide a benefit \(b>c\) to its partner. The population evolves according to a Wright--Fisher model with mutation in both traits: strategy with probability \(u\) and phenotype with probability \(v\). We have derived an explicit threshold \(\beta_n(\alpha)\) that the benefit-to-cost ratio must exceed for weak selection to favor the abundance of cooperation in the stationary equilibrium. For \(N\) large enough, this threshold can be expressed in terms of a generalized Laplace transform.

Our first main result concerns the effect of the discrimination parameter \(\alpha\). When \(\alpha \to 0\), a cooperator helps any other individual with probability approaching \(1\), so no phenotypic assortment is generated. In this case, \(\beta_n(0^+) = \infty\), recovering the known result that a well-mixed population never favors cooperation~\citep{LehmannRoussett2010}. As \(\alpha\) increases from zero, cooperative acts become increasingly concentrated on phenotypically similar partners, and \(\beta_n(\alpha)\) decreases. The selective advantage of cooperation rises because assortment becomes more efficient, even though the set of likely recipients shrinks. This means that selection can favor the abundance of cooperation, even when recognition is not perfect. The threshold attains its minimum in the limit \(\alpha \to \infty\), which corresponds to the exact-matching model of~\citep{Antal2009c,Kroumi2015}. Thus, sharper discrimination is the most favorable scenario for the abundance of cooperation. This agrees with the idea that recognition is most effective when cues are sufficiently reliable~\citep{Grafen1990,RoussetRoze2007}.

The effect of the phenotype-space dimension \(n\) is geometric in origin. A single mutation displaces a lineage along one coordinate axis while leaving all other components unchanged. Two lineages that accumulate mutations along different axes diverge and move out of one another's interaction range, whereas those that remain close continue to interact. Adding a dimension therefore reduces the overlap between unrelated lineages and sharpens the phenotypic filtering of social interactions. In the exact-matching model~\citep{Kroumi2015}, the authors observed numerically that the threshold decreases with \(n\) and suggested that this reflects a general pattern. We have proved that this is indeed true for all \(\alpha>0\). The decrease from \(n=1\) to \(n=2\) is especially pronounced; beyond \(n=2\), the reduction is smaller but remains strictly positive for every integer \(n\), as shown in Figure~4. This suggests that recognition systems combining multiple independent cues may support more cooperative interactions and thereby strengthen assortment among socially compatible individuals. Richer phenotype spaces can therefore sharpen partner discrimination beyond what any single cue can achieve~\citep{Gruenheit2017}.

The mutation rates, the phenotype mutation rate \(\nu=\lim_{N\rightarrow\infty}Nv\) and the strategy mutation rate \(\mu=\lim_{N\rightarrow\infty}Nu\), act in opposite directions. When \(\nu\) is small, lineages remain clustered within a few phenotypic classes for long periods. Then, cooperators and defectors occupy the same phenotypic neighborhoods, and the threshold diverges as
\[
\beta_n(\alpha) \sim \frac{(1+2\mu)(3+2\mu)}{4(1+\mu)(1-e^{-\alpha})}\,\frac{1}{\nu},
\,
\]
as \(\nu \to 0^+\). Therefore, the population behaves much like a well-mixed one. By contrast, when \(\nu\) is large and \(n \geq 2\), lineages spread across the lattice and interactions become concentrated among those that remain phenotypically close. In this regime, \(\beta_n(\alpha) \to 1\), so \(b>c\) becomes sufficient regardless of \(\alpha\) and \(\mu\). A high phenotype mutation rate keeps groups small and homogeneous, so the strategy with the higher payoff against itself, namely cooperation, is favored. Tag-based models reach a similar conclusion through a different route, showing that rapid turnover of recognition markers allows cooperators to escape exploitation and reestablish under new tags~\citep{Jansen2006,TraulsenNowak2007}. The one-dimensional case is exceptional. Even as \(\nu \to \infty\), the threshold does not fall to \(1\) but instead converges to 
\begin{equation*}
\lim_{\nu\to\infty}\beta_1(\alpha) = \frac{(1+2\mu)^2 + \mu\sqrt{3+2\mu} - (1+\mu)\sqrt{1+2\mu}}{\mu\sqrt{3+2\mu} + (1+\mu)\sqrt{1+2\mu} - (1+2\mu)}>1.
\end{equation*}
It grows with \(\mu\) and diverges as \(\mu \to \infty\). This means that a one-dimensional lattice does not separate diverging lineages as effectively as higher-dimensional spaces, so phenotypic dispersion alone is not sufficient to favor the abundance of cooperation.

The strategy mutation rate \(\mu\) has the opposite effect. Large values of \(\mu\) continuously reintroduce defectors into phenotypic clusters composed entirely of cooperators, while also reintroducing cooperators into clusters of defectors.  This weakens the association between phenotype and strategy on which the assortment mechanism depends. Because defectors can invade cooperative groups more effectively than cooperators can invade groups of defectors, selection tends to favor the strategy with the higher payoff against the other, namely defection. In this case, the threshold grows linearly with \(\mu\) and diverges, making it impossible for selection to favor cooperation regardless of the value of \(b/c\). This effect becomes stronger as \(\mu\) increases. When \(\mu\) is small, cooperative clusters persist long enough for phenotypic assortment to protect them. Recognition succeeds not merely because individuals carry phenotypic cues, but because those cues remain predictive of behavior across generations~\citep{Axelrod2004,Jansen2006,GardnerWest2010}. The most favorable regime for cooperation combines a high phenotype mutation rate with a low strategy mutation rate.

Finally, note that the exact-matching model is a limiting case that leads to the lowest possible threshold. Weakening the exactness of recognition raises \(\beta_n(\alpha)\), but preserves the qualitative structure of the model. The threshold still decreases with \(\alpha\), decreases with \(n\), diverges when \(\nu\) is small, and diverges when \(\mu\) is large. What changes is quantitative rather than qualitative. Imperfect recognition requires a larger benefit-to-cost ratio for cooperation to be favored, but it does not change how the threshold depends on the parameters. The conclusions of exact-matching models therefore do not depend solely on the assumption of binary recognition. Instead, they reflect the geometry of phenotype space and the way lineages separate within it.

\appendix

\section{\label{AppendixA}Calculation of the identity measures}
In this appendix, we will analyze the model in the large-population limit ($N \to \infty$) with vanishing mutation probabilities ($u, v \to 0$), while keeping the mutation rates $\mu = Nu$ and $\nu = Nv$ constant. In this case, the ancestry of individuals is described by the Kingman coalescent, where the rescaled coalescence time $T$ for two lineages has the exponential density $f_1(\tau) = e^{-\tau}$ for $\tau > 0$. Conditional on \(T=\tau\), strategy mutations and phenotype mutations on the ancestral lines are independent. 

The number of strategy mutations on the ancestral branches of individuals $j$ and $k$ up to their most recent common ancestor is a Poisson random variable with  mean $2\mu\tau$. Then, they share the same strategy if there is no mutation on their ancestral branches which occurs with probability $e^{-2\mu\tau}$ or if the last mutation moving from individual $j$ to their MRCA and then to individual $k$ will lead to the same type as individual $j$ which occurs with probability $(1-e^{-2\mu\tau})\times\frac{1}{2}$. Therefore, we have 
\begin{equation}\label{eq:pair-strategy-self}
\mathbb P_0\bigl(S(k)=S(j)\mid T=\tau\bigr)=\frac{1+e^{-2\mu\tau}}{2}.
\end{equation}

Now, we are interested in
\begin{equation}
    \psi_n(\alpha,\tau):=\mathbb E_0\!\left[e^{-\alpha\|\mathbf y(k)-\mathbf y(j)\|_1}\mid T=\tau\right].
\end{equation}
Similarly, conditional on $T=\tau$, the number of phenotype mutations in one ancestral line is a Poisson random variable with mean $\nu\tau/n$ on each coordinate \(i\in\{1,\dots,n\}\). Define the coordinate difference
\[
d_i(k,j):=y_i(k)-y_i(j).
\]
Along one ancestral line, mutations in the \(i\)-th coordinate occur at a total rate \(\nu/n\), split equally between \(+1\) and \(-1\) jumps. Since two lineages are involved, the total number of \(+1\) jumps contributing to the difference \(d_i(k,j)=y_i(k)-y_i(j)\) is Poisson with parameter \(\nu\tau/n\), and the same holds for the total number of \(-1\) jumps. Hence, $d_i(k,j)$
has a centered Skellam distribution, where
\[
\mathbb P(d_i(k,j)=m\mid T=\tau)=e^{-\sigma}I_{|m|}(\sigma).
\]
Here, $I_m$ is the modified Bessel function of the first kind given by
$$
I_m(x)=\sum_{\ell=0}^{\infty}\frac{1}{\ell!(2\ell+|m|)!}\left(\frac{x}{2}\right)^{2\ell+|m|}
$$
and $\sigma=2\nu\tau/n$.
Therefore, we obtain
\begin{equation}\label{eq:one-coordinate-kernel-self}
\mathbb E_0\!\left[e^{-\alpha |d_i(k,j)|}\mid T=\tau\right]
=\sum_{m=-\infty}^{\infty}e^{-\alpha|m|}e^{-\sigma}I_{|m|}(\sigma)=
e^{-\sigma}\Phi_\alpha(\sigma),
\end{equation}
where
\begin{equation}\label{eq:Phi-alpha-def-self}
\Phi_\alpha(\sigma)=\sum_{m=-\infty}^{\infty}e^{-\alpha|m|}I_{|m|}(\sigma)
=I_0(\sigma)+2\sum_{m=1}^{\infty}e^{-\alpha m}I_m(\sigma).
\end{equation}
Since the coordinate differences \(d_1(k,j),\dots,d_n(k,j)\) are independent and
\(
\|\mathbf y(k)-\mathbf y(j)\|_1=\sum_{i=1}^n |d_i(k,j)|,
\)
we obtain
\begin{equation}\label{eq:psi-alpha-self}
\psi_n(\alpha,\tau)=\prod_{i=1}^n \mathbb E_0\!\left[e^{-\alpha|d_i(k,j)|}\mid T=\tau\right]
=\Bigl[e^{-\sigma}\Phi_\alpha(\sigma)\Bigr]^n.
\end{equation}

The function \(\Phi_\alpha\) also admits a useful integral representation. Indeed, using
\[
I_m(\sigma)=\frac1\pi\int_0^\pi e^{\sigma\cos\theta}\cos(m\theta)\,d\theta
\]
and the Poisson kernel identity
\[
1+2\sum_{m=1}^\infty r^m\cos(m\theta)=\frac{1-r^2}{1-2r\cos\theta+r^2},
\]
with \(r=e^{-\alpha}\), we get
\begin{equation}\label{eq:Phi-alpha-integral-self}
\begin{split}
\Phi_\alpha(\sigma)&=\frac1\pi\int_0^\pi e^{\sigma\cos\theta}\left(1+2\sum_{m=1}^\infty e^{-\alpha m}\cos(m\theta)\right)\,d\theta\\
&=\frac{1-e^{-2\alpha}}{\pi}\int_0^\pi\frac{e^{\sigma\cos\theta}}{1-2e^{-\alpha}\cos\theta+e^{-2\alpha}}\,d\theta.
\end{split}
\end{equation}
Using the tower property in Eqs. \eqref{eq:pair-strategy-self} and \eqref{eq:psi-alpha-self} yields
\begin{subequations}
\begin{align}
Z_n(\alpha)&=\int_0^\infty \psi_n(\alpha,\tau)e^{-\tau}\,d\tau =\frac{n}{2\nu}\,\mathcal{L}_n\!\left(\alpha,\frac{1}{2\nu}\right),\label{eq:Zalpha-Gtilde-self}\\
G_n(\alpha)&=\int_0^\infty\psi_n(\alpha,\tau)\,\frac{1+e^{-2\mu\tau}}{2}\,e^{-\tau}\,d\tau=\frac{n}{4\nu}\left[\mathcal{L}_n\!\left(\alpha,\frac{1}{2\nu}\right)+\mathcal{L}_n\!\left(\alpha,\frac{1+2\mu}{2\nu}\right)\right],\label{eq:Galpha-Gtilde-self}
\end{align}
\end{subequations}
where $\mathcal{L}_n$ is a generalized Laplace transform given by
\begin{equation}\label{eq:Gtilde-def-self}
\mathcal{L}_n(\alpha,x)=\int_0^\infty [\Phi_\alpha(\sigma)]^n e^{-n(1+x)\sigma}\,d\sigma.
\end{equation}

To compute \(H_n(\alpha)\), let \(\tau_3\) be the first coalescence time among three sampled lineages of individuals $k$, $j$ and $\ell$ and \(\tau_2\) the additional time until the remaining lineage coalesces with the first pair. Note that \((\tau_3,\tau_2)\) has density
\[
f_2(\tau_3,\tau_2)=3e^{-(3\tau_3+\tau_2)},
\]
for $\tau_3,\tau_2>0$.
There are three equiprobable scenarios. Scenario $1$ is when \(k\) and \(j\) coalesce first, in which the phenotype factor between \(k\) and \(j\) is \(\psi_n(\alpha,\tau_3)\), while the strategy identity between \(j\) and \(\ell\) is conditioned on the total time \(\tau_3+\tau_2\). Hence, we have
\begin{equation}
\begin{split}
H_{n}^{(1)}(\alpha):&=\mathbb E_0\!\left[\varphi(k,j)\,\mathbf 1_{\{S(j)=S(\ell)\}}\Big| \text{Scenario }1\right]\\
&=\int_0^\infty\int_0^\infty3e^{-(3\tau_3+\tau_2)}\psi_n(\alpha,\tau_3)\frac{1+e^{-2\mu(\tau_3+\tau_2)}}{2}\,d\tau_2\,d\tau_3\\
&=\frac32\int_0^\infty \psi_n(\alpha,\tau)\left(e^{-3\tau}+\frac{1}{1+2\mu}e^{-(3+2\mu)\tau}\right)\,d\tau.
\end{split}
\end{equation}
Scenario $2$ is when \(j\) and \(\ell\) coalesce first. In this case, the strategy factor between $j$ and $\ell$ uses \(\tau_3\), but the phenotype factor between \(k\) and \(j\) uses the total time \(\tau_3+\tau_2\):
\begin{equation}
\begin{split}
H_{n}^{(2)}(\alpha):&=\mathbb E_0\!\left[\varphi(k,j)\,\mathbf 1_{\{S(j)=S(\ell)\}}\Big| \text{Scenario }2\right]\\
&=\int_0^\infty\int_0^\infty3e^{-(3\tau_3+\tau_2)}\psi_n(\alpha,\tau_3+\tau_2)\frac{1+e^{-2\mu\tau_3}}{2}\,d\tau_2\,d\tau_3\\
&=\frac34\int_0^\infty \psi_n(\alpha,\tau)\left(\frac{2+\mu}{1+\mu}e^{-\tau}-e^{-3\tau}-\frac{1}{1+\mu}e^{-(3+2\mu)\tau}\right)\,d\tau.
\end{split}
\end{equation}
Scenario 3 is when \(k\) and \(\ell\) coalesce first, under which both the phenotype factor between \(k\) and \(j\) and the strategy identity between \(j\) and \(\ell\) involve the total time \(\tau_3+\tau_2\), leading to
\begin{equation}
\begin{split}
H_{n}^{(3)}(\alpha):&=\mathbb E_0\!\left[\varphi(k,j)\,\mathbf 1_{\{S(j)=S(\ell)\}}\Big| \text{Scenario }3\right]\\
&=\int_0^\infty\int_0^\infty3e^{-(3\tau_3+\tau_2)}\psi_n(\alpha,\tau_3+\tau_2)\frac{1+e^{-2\mu(\tau_3+\tau_2)}}{2}\,d\tau_2\,d\tau_3\\
&=\frac34\int_0^\infty \psi_n(\alpha,\tau)\left(e^{-\tau}+e^{-(1+2\mu)\tau}-e^{-3\tau}-e^{-(3+2\mu)\tau}\right)\,d\tau.
\end{split}
\end{equation}
Summing these three expressions and dividing by \(3\), we obtain
\begin{align}
H_n(\alpha)&=\frac{H_{n}^{(1)}(\alpha)+H_{n}^{(2)}(\alpha)+H_{n}^{(3)}(\alpha)}{3}\notag\\
&=\frac14\int_0^\infty \psi_n(\alpha,\tau)\left(\frac{3+2\mu}{1+\mu}e^{-\tau}+e^{-(1+2\mu)\tau}-\frac{\mu(3+2\mu)}{(1+\mu)(1+2\mu)}e^{-(3+2\mu)\tau}\right)\,d\tau \notag\\
&=\frac{n}{8\nu}\left[\frac{3+2\mu}{1+\mu}\,\mathcal{L}_n\!\left(\alpha,\frac{1}{2\nu}\right)+\mathcal{L}_n\!\left(\alpha,\frac{1+2\mu}{2\nu}\right)-
\frac{\mu(3+2\mu)}{(1+\mu)(1+2\mu)}\,\mathcal{L}_n\!\left(\alpha,\frac{3+2\mu}{2\nu}\right)\right].\label{eq:Halpha-Gtilde-self}
\end{align}

Similarly, the three-point identity
\[
\eta_n(\alpha)
=
\mathbb E_0\!\left[\varphi(k,j)\,\mathbf 1_{\{S(k)=S(j)=S(\ell)\}}\right]
\]
is obtained from the same three-lineage coalescent decomposition. Conditional on the coalescence times \((\tau_3,\tau_2)\), the probability that individuals $k,\,j,\,\ell$ carry the same strategy is
\[
\mathbb P_0\bigl(S(k)=S(j)=S(\ell)\mid \tau_3,\tau_2\bigr)
=
\frac14\Bigl(1+e^{-2\mu\tau_3}+2e^{-2\mu(\tau_3+\tau_2)}\Bigr).
\]
By analogy, if \(k\) and \(j\) coalesce first, then the phenotype factor is \(\psi_n(\alpha,\tau_3)\), whereas in the other two scenarios it is \(\psi_n(\alpha,\tau_3+\tau_2)\). Therefore,
\[
\eta_n(\alpha)=\frac{\eta_n^{(1)}(\alpha)+\eta_n^{(2)}(\alpha)+\eta_n^{(3)}(\alpha)}{3},
\]
with
\begin{align*}
\eta_n^{(1)}(\alpha)
&=
\int_0^\infty\!\!\int_0^\infty
3e^{-(3\tau_3+\tau_2)}
\psi_n(\alpha,\tau_3)
\frac14\Bigl(1+e^{-2\mu\tau_3}+2e^{-2\mu(\tau_3+\tau_2)}\Bigr)\,d\tau_2\,d\tau_3,\\
\eta_n^{(2)}(\alpha)
&=
\int_0^\infty\!\!\int_0^\infty
3e^{-(3\tau_3+\tau_2)}
\psi_n(\alpha,\tau_3+\tau_2)
\frac14\Bigl(1+e^{-2\mu\tau_3}+2e^{-2\mu(\tau_3+\tau_2)}\Bigr)\,d\tau_2\,d\tau_3,\\
\eta_n^{(3)}(\alpha)&=\eta_n^{(2)}(\alpha).
\end{align*}
Carrying out the \(\tau_2\)-integration in the first term and using the change of variables \(\tau=\tau_3+\tau_2\) in the second and third terms, exactly as for \(H_n(\alpha)\), yields
\begin{align}
\eta_n(\alpha)
&=
\frac{n}{8\nu}\left[
\frac{2+\mu}{1+\mu}\,\mathcal{L}_n\!\left(\alpha,\frac{1}{2\nu}\right)
+2\mathcal{L}_n\!\left(\alpha,\frac{1+2\mu}{2\nu}\right)
-\frac{\mu(3+2\mu)}{(1+\mu)(1+2\mu)}\,\mathcal{L}_n\!\left(\alpha,\frac{3+2\mu}{2\nu}\right)
\right].
\label{eq:eta-largeN-self}
\end{align}

\section{\label{appendixB}Proof of Result \ref{thm:alpha-monotonicity-section}}
Write 
\begin{equation}\label{eq37}
\Phi_\alpha(\sigma)=\int_0^\pi e^{\sigma\cos\theta}\, h_\alpha(\theta)\,d\theta,
\end{equation}
where
\begin{equation}\label{eq:poisson-density-alpha}
h_\alpha(\theta)
=
\frac{1-e^{-2\alpha}}{\pi(1-2e^{-\alpha}\cos\theta+e^{-2\alpha})}.
\end{equation}

\begin{lemma}\label{lem:score-alpha-section}
Define the function
\[
w_\alpha(\theta)
:=
\frac{\partial}{\partial\alpha} \log h_\alpha(\theta).
\]
Then \(w_\alpha(\theta)\) is strictly increasing in \(\theta\in(0,\pi)\).
\end{lemma}

\begin{proof}
From \eqref{eq:poisson-density-alpha}, we have
\[
\log h_\alpha(\theta)
=
\log\left(1-e^{-2\alpha}\right)-\log\pi-\log\left(1-2e^{-\alpha}\cos\theta+e^{-2\alpha}\right).
\]
Therefore
\[
w_\alpha(\theta)=\frac{2e^{-2\alpha}}{1-e^{-2\alpha}}+\frac{2e^{-\alpha}(e^{-\alpha}-\cos\theta)}{1-2e^{-\alpha}\cos\theta+e^{-2\alpha}}.
\]
Differentiating with respect to \(\theta\)
gives
\[
\frac{\partial}{\partial\theta}w_\alpha(\theta)
=
\frac{2e^{-\alpha}(1-e^{-2\alpha})\sin\theta}{\bigl(1-2e^{-\alpha}\cos\theta+e^{-2\alpha}\bigr)^2}>0,
\]
for \(\theta\in(0,\pi)\). This completes the proof.
\end{proof}

Define
\(
F_\alpha(\sigma)=[\Phi_\alpha(\sigma)]^n e^{-n\sigma},
\)
 such that
\begin{equation}\label{eq9}
\mathcal{L}_n(\alpha,x)=\int_0^\infty F_\alpha(\sigma)e^{-nx\sigma}\,d\sigma.
\end{equation}

\begin{lemma}\label{lem:cross-partial-alpha-section}
For every \(\alpha>0\), the function
\(
\sigma\longmapsto \dfrac{\partial}{\partial\alpha}\log F_\alpha(\sigma)
\)
is strictly decreasing on \((0,\infty)\).
\end{lemma}

\begin{proof}
Define the probability measure \(\mathbb P_{\alpha,\sigma}\) on \([0,\pi]\) by
\[
d\mathbb P_{\alpha,\sigma}(\theta)
=
\frac{e^{\sigma\cos\theta}h_\alpha(\theta)}{\Phi_\alpha(\sigma)}\,d\theta.
\]
Using \eqref{eq37}, we have
\begin{equation}\label{eq6}
\frac{\partial}{\partial\alpha} \log \Phi_\alpha(\sigma)
=
\frac{\int_0^\pi e^{\sigma\cos\theta}\,\frac{\partial}{\partial\alpha} h_\alpha(\theta)\,d\theta}
{\Phi_\alpha(\sigma)}
=
\int_0^\pi w_\alpha(\theta)\,d\mathbb P_{\alpha,\sigma}(\theta).
\end{equation}
Differentiating the last identity with respect to \(\sigma\) yields
\begin{align*}
\frac{\partial^2}{\partial \sigma\partial\alpha}\log \Phi_\alpha(\sigma)&=  \frac{\partial}{\partial \sigma}\left[
\int_0^\pi w_\alpha(\theta)\,\frac{e^{\sigma\cos\theta}h_\alpha(\theta)}{\Phi_\alpha(\sigma)}\,d\theta
\right] \\
&=\int_0^\pi w_\alpha(\theta)\left[\cos\theta\,\frac{e^{\sigma\cos\theta}h_\alpha(\theta)}{\Phi_\alpha(\sigma)}-\frac{\frac{\partial}{\partial \sigma}\Phi_\alpha(\sigma)}{\Phi_\alpha(\sigma)}\frac{e^{\sigma\cos\theta}h_\alpha(\theta)}{\Phi_\alpha(\sigma)}\right]\,d\theta\\
&=\int_0^\pi w_\alpha(\theta)\cos\theta\,d\mathbb P_{\alpha,\sigma}(\theta)-\frac{\frac{\partial}{\partial \sigma}\Phi_\alpha(\sigma)}{\Phi_\alpha(\sigma)}\int_0^\pi w_\alpha(\theta)d\mathbb P_{\alpha,\sigma}(\theta).
\end{align*}
On the other hand, we have
\begin{equation}
 \frac{\frac{\partial}{\partial \sigma}\Phi_\alpha(\sigma)}{\Phi_\alpha(\sigma)}= \frac{\frac{\partial}{\partial \sigma}\int_0^\pi e^{\sigma\cos\theta}h_\alpha(\theta)\,d\theta  }{\Phi_\alpha(\sigma)}
 =\int_0^\pi \cos\theta\,\frac{e^{\sigma\cos\theta}h_\alpha(\theta)}{\Phi_\alpha(\sigma)}\,d\theta=\int_0^\pi \cos(\theta)d\mathbb P_{\alpha,\sigma}(\theta).
\end{equation}
This shows that
\[
\frac{\partial^2}{\partial \sigma\partial\alpha}\log \Phi_\alpha(\sigma)
=
\Cov_{\mathbb P_{\alpha,\sigma}}\!\bigl(w_\alpha(\theta),\cos\theta\bigr).
\]
By Lemma~\ref{lem:score-alpha-section}, \(w_\alpha(\theta)\) is strictly increasing in \(\theta\). As \(\cos\theta\) is strictly decreasing on \((0,\pi)\), then Chebyshev's covariance inequality on the line gives
\[
\Cov_{\mathbb P_{\alpha,\sigma}}\!\bigl(w_\alpha(\theta),\cos\theta\bigr)<0.
\]
The inequality is strict since neither function is \(\mathbb P_{\alpha,\sigma}\)-a.s.\ constant. Noting that
\[
\frac{\partial^2}{\partial \sigma\partial\alpha}\log F_\alpha(\sigma)
=
n\,\frac{\partial^2}{\partial \sigma\partial\alpha}\log \Phi_\alpha(\sigma),
\]
completes the proof.
\end{proof}
\begin{lemma}\label{lem:ordering-alpha-section}
For \(a>0\), define
\[
X(a)
:=
\frac{\partial}{\partial\alpha} \log \mathcal{L}_n(\alpha,a).
\]
Then \(X(a)\) is strictly increasing in \(a\). In addition, we have
\(
X(a)<0.
\)
\end{lemma}

\begin{proof}
Define the probability measure \(Q_{\alpha,a}\) on \((0,\infty)\) by
\[
dQ_{\alpha,a}(\sigma)
=
\frac{F_\alpha(\sigma)e^{-na\sigma}}{\mathcal{L}_n(\alpha,a)}\,d\sigma.
\]
Then
\[
X(a)
=
\int_0^\infty \frac{\partial}{\partial\alpha} \log F_\alpha(\sigma)\,dQ_{\alpha,a}(\sigma).
\]
Fix \(a'>a\). The Radon--Nikodym derivative of \(Q_{\alpha,a'}\) with respect to \(Q_{\alpha,a}\) is
\[
\frac{dQ_{\alpha,a'}}{dQ_{\alpha,a}}(\sigma)
=
\frac{\mathcal{L}_n(\alpha,a)}{\mathcal{L}_n(\alpha,a')}e^{-n(a'-a)\sigma},
\]
which is strictly decreasing in \(\sigma\). Hence \(Q_{\alpha,a'}\) is first-order stochastically smaller than \(Q_{\alpha,a}\). By Lemma~\ref{lem:cross-partial-alpha-section}, the function
\[
\sigma\longmapsto \frac{\partial}{\partial\alpha} \log F_\alpha(\sigma)
\]
is strictly decreasing. Therefore
\[
X(a')
=
\mathbb E_{Q_{\alpha,a'}}\!\left[\frac{\partial}{\partial\alpha}\log F_\alpha(\sigma)\right]
>
\mathbb E_{Q_{\alpha,a}}\!\left[\frac{\partial}{\partial\alpha}\log F_\alpha(\sigma)\right]
=
X(a).
\]
This proves that \(X\) is strictly increasing. Finally, \(\mathcal{L}_n(\alpha,a)\) is strictly decreasing in \(\alpha\) for every \(a>0\), because \(\Phi_\alpha(\sigma)\) is strictly decreasing in \(\alpha\) for every \(\sigma>0\). Hence \(X(a)=\frac{\partial}{\partial\alpha} \log \mathcal{L}_n(\alpha,a)<0\).
\end{proof}

Let us introduce the shorthand notations $a_1=\dfrac{1}{2\nu}$, $a_2=\dfrac{1+2\mu}{2\nu}$ and $a_3=\dfrac{3+2\mu}{2\nu}$. 

\begin{theorem}
For every \(n\ge 1\), \(\mu>0\), \(\nu>0\), and \(\alpha>0\), we have
\[
\frac{\partial}{\partial\alpha}\beta_n(\alpha)<0.
\]
\end{theorem}

\begin{proof}
Write 
\begin{equation}\label{eq:beta-alpha-section}
\beta_n(\alpha)
=\frac{M_1(\alpha)}{M_2(\alpha)},
\end{equation}
where
\begin{align*}
M_1(\alpha)&=(1+2\mu)^2\mathcal{L}_n(\alpha,a_1)+\mu(3+2\mu)\mathcal{L}_n(\alpha,a_3)-(1+\mu)(1+2\mu)\mathcal{L}_n(\alpha,a_2),\\
M_2(\alpha)&=-(1+2\mu)\mathcal{L}_n(\alpha,a_1)+\mu(3+2\mu)\mathcal{L}_n(\alpha,a_3)+(1+\mu)(1+2\mu)\mathcal{L}_n(\alpha,a_2).
\end{align*}
To show that $\alpha\rightarrow\beta_n(\alpha)$ is decreasing, it suffices to show that
\begin{equation}\label{eq38}
M_2(\alpha)\frac{\partial}{\partial\alpha} M_1(\alpha)-M_1(\alpha)\frac{\partial}{\partial\alpha} M_2(\alpha)<0,
\end{equation}
for any $\alpha>0$. 
Differentiating, expanding and rearranging the different terms, we get
\begin{equation}\label{eq4}
\begin{split}
 & M_2(\alpha)\frac{\partial}{\partial\alpha} M_1(\alpha)-M_1(\alpha)\frac{\partial}{\partial\alpha} M_2(\alpha)  \\
  =&2\mu(1+\mu)(1+2\mu)^2
\left[\mathcal{L}_n(\alpha,a_2) \frac{\partial}{\partial\alpha} \mathcal{L}_n(\alpha,a_1)-\mathcal{L}_n(\alpha,a_1) \frac{\partial}{\partial\alpha} \mathcal{L}_n(\alpha,a_2)\right]\\
&+
2\mu(1+2\mu)(1+\mu)(3+2\mu)
\left[\mathcal{L}_n(\alpha,a_3) \frac{\partial}{\partial\alpha} \mathcal{L}_n(\alpha,a_1)-\mathcal{L}_n(\alpha,a_1) \frac{\partial}{\partial\alpha} \mathcal{L}_n(\alpha,a_3)\right]\\
&-2\mu(1+\mu)(1+2\mu)(3+2\mu)\left[\mathcal{L}_n(\alpha,a_3) \frac{\partial}{\partial\alpha} \mathcal{L}_n(\alpha,a_2)-\mathcal{L}_n(\alpha,a_2) \frac{\partial}{\partial\alpha} \mathcal{L}_n(\alpha,a_3)\right].
\end{split}
\end{equation}
On the other hand, we have \[
X(a)=\frac{\partial}{\partial\alpha} \log \mathcal{L}_n(\alpha,a)=\frac{\frac{\partial}{\partial\alpha} \mathcal{L}_n(\alpha,a)}{\mathcal{L}_n(\alpha,a)},
\]
Then, we obtain
$$\frac{\partial}{\partial\alpha} \mathcal{L}_n(\alpha,a)= X(a)\mathcal{L}_n(\alpha,a).$$
Inserting this expression in \eqref{eq4} yields
\begin{equation}\label{eq5}
\begin{split}
  &M_2(\alpha)\frac{\partial}{\partial\alpha} M_1(\alpha)-M_1(\alpha)\frac{\partial}{\partial\alpha} M_2(\alpha)  \\
  =&2\mu(1+\mu)(1+2\mu)^2\mathcal{L}_n(\alpha,a_2)\mathcal{L}_n(\alpha,a_1)
\left[ X(a_1)-X(a_2)\right]\\
&+
2\mu(1+\mu)(1+2\mu)(3+2\mu)\mathcal{L}_n(\alpha,a_1)\mathcal{L}_n(\alpha,a_3)
\left[X(a_1)-X(a_3)\right]\\
&-2\mu(1+\mu)(1+2\mu)(3+2\mu)\mathcal{L}_n(\alpha,a_3) \mathcal{L}_n(\alpha,a_2)\left[X(a_2)-X(a_3)\right]\\
=&-2\mu(1+\mu)(1+2\mu)\mathcal{L}_n(\alpha,a_1)\Bigl[
(1+2\mu)\mathcal{L}_n(\alpha,a_2)
+(3+2\mu)\mathcal{L}_n(\alpha,a_3)
\Bigr]
\Delta_{12}\\
&
-2\mu(1+\mu)(1+2\mu)(3+2\mu)\mathcal{L}_n(\alpha,a_3)\left[
\mathcal{L}_n(\alpha,a_1)
-
 \mathcal{L}_n(\alpha,a_2)
\right]\Delta_{23},
\end{split}
\end{equation}
where $\Delta_{12}=X(a_2)-X(a_1)$ and $\Delta_{23}=X(a_3)-X(a_2)$. By Lemma~\ref{lem:ordering-alpha-section}, note that $\Delta_{12}>0$ and $\Delta_{23}>0$.
On the other hand, since \(a_1<a_2<a_3\) and \(x\mapsto \mathcal{L}_n(\alpha,x)\) is strictly decreasing, we have
\[
\mathcal{L}_n(\alpha,a_1)>\mathcal{L}_n(\alpha,a_2)>\mathcal{L}_n(\alpha,a_3)>0.
\]
Substituting these inequalities into \eqref{eq5} shows the inequality in \eqref{eq38}. 
\end{proof}


\section{\label{AppendixC}Proof of Result \ref{thm:n-monotonicity-section}}
Using the change of variables \(u=n\sigma\) in the integral form in Eq. \eqref{eq15}, we have
\begin{equation}\label{eq:Gtilde-Laplace-n-section}
\mathcal{L}_n(\alpha,x)=\frac1n\int_0^\infty \widetilde f_n(u)e^{-xu}\,du=\frac1n\,\mathcal L[\widetilde f_n](x),
\end{equation}
where 
\begin{equation}\label{eq:f-tilde-n-section}
\widetilde f_n(u)=\bigl[\Phi_\alpha(u/n)\bigr]^ne^{-u}.
\end{equation}
The expression in \eqref{eq:Gtilde-Laplace-n-section} is extended for any real $n>0$. 
Thus, the entire dependence of $\mathcal{L}_n(\alpha,x)$ on \(n\) is now transferred to the family of functions \(\widetilde f_n\).
Define
$$\phi_\alpha(\sigma)=\log \Phi_\alpha(\sigma)-\sigma.$$
\begin{lemma}\label{lem:phi-convex-n-section}
For every \(\alpha>0\) and $\sigma>0$, we have $\dfrac{\partial^2}{\partial \sigma^2}\phi_\alpha(\sigma)>0$.
\end{lemma}

\begin{proof}
Similarly to \eqref{eq6}, we have
\[
\frac{\partial}{\partial \sigma}\log \Phi_\alpha(\sigma)
=
\int_0^\pi \cos\theta\,d\mathbb P_{\alpha,\sigma}(\theta),
\]
from which, we get
\begin{align*}
\frac{\partial^2}{\partial \sigma^2}\log \Phi_\alpha(\sigma)&=  \frac{\partial}{\partial \sigma}\left[
\int_0^\pi \cos(\theta)\,\frac{e^{\sigma\cos\theta}h_\alpha(\theta)}{\Phi_\alpha(\sigma)}\,d\theta
\right] \\
&=\int_0^\pi \cos(\theta)\left[\cos\theta\,\frac{e^{\sigma\cos\theta}h_\alpha(\theta)}{\Phi_\alpha(\sigma)}-\frac{\frac{\partial}{\partial \sigma}\Phi_\alpha(\sigma)}{\Phi_\alpha(\sigma)}\frac{e^{\sigma\cos\theta}h_\alpha(\theta)}{\Phi_\alpha(\sigma)}\right]\,d\theta\\
&=\int_0^\pi \cos^2\theta\,d\mathbb P_{\alpha,\sigma}(\theta)-\frac{\frac{\partial}{\partial \sigma}\Phi_\alpha(\sigma)}{\Phi_\alpha(\sigma)}\int_0^\pi \cos(\theta)d\mathbb P_{\alpha,\sigma}(\theta)\\
&=\Var_{\mathbb P_{\alpha,\sigma}}(\cos\theta).
\end{align*}
Since \(\cos\theta\) is not almost surely constant under \(\mathbb P_{\alpha,\sigma}\), the variance is strictly positive. Hence
\[
\frac{\partial^2}{\partial \sigma^2}\phi_\alpha(\sigma)
=
\frac{\partial^2}{\partial \sigma^2}\log \Phi_\alpha(\sigma)
=
\Var_{\mathbb P_{\alpha,\sigma}}(\cos\theta)>0.
\]
\end{proof}
The next step is to understand how \(\widetilde f_n\) changes with \(n\). For that, define
\begin{equation}\label{eq:g-function-n-section}
g_\alpha(\sigma):=\phi_\alpha(\sigma)-\sigma\frac{\partial}{\partial \sigma}\phi_\alpha(\sigma).
\end{equation}

\begin{lemma}\label{lem:g-negative-n-section}
For every \(\alpha>0\), the function \(\sigma\mapsto g_\alpha(\sigma)\) is strictly decreasing on \((0,\infty)\). In particular, \(g_\alpha(\sigma)<0\) for every \(\sigma>0\).
\end{lemma}

\begin{proof}
Differentiating \eqref{eq:g-function-n-section} gives
\[
\frac{\partial}{\partial \sigma} g_\alpha(\sigma)
=
-\sigma\frac{\partial^2}{\partial \sigma^2}\phi_\alpha(\sigma).
\]
By Lemma~\ref{lem:phi-convex-n-section}, \(\dfrac{\partial^2}{\partial \sigma^2}\phi_\alpha(\sigma)>0\) for \(\sigma>0\), so that \(\dfrac{\partial}{\partial \sigma} g_\alpha(\sigma)<0\). This shows that 
$\sigma\rightarrow g_\alpha(\sigma)$ is strictly decreasing. On the other hand, 
\[
g_\alpha(0)=\phi_\alpha(0)=\log\Phi_\alpha(0)=0,
\]
because 
\[
\Phi_\alpha(0)=\frac{1 - e^{-2\alpha}}{\pi} \int_0^\pi \frac{1}{1 - 2e^{-\alpha}\cos\theta + e^{-2\alpha}}  d\theta=1.
\]
Therefore \(g_\alpha(\sigma)<0\) for all \(\sigma>0\).

\end{proof}
Using Eq. \eqref{eq:f-tilde-n-section}, we have
\[
\log \widetilde f_n(u)=n\phi_\alpha(u/n).
\]
Differentiating this expression with respect to \(n\) yields
\begin{equation}\label{eq7}
\frac{\partial}{\partial n}\log \widetilde f_n(u)
=
\phi_\alpha(u/n)-\frac{u}{n}\frac{\partial}{\partial \sigma}\phi_\alpha(u/n)
=
g_\alpha(u/n).
\end{equation}

\begin{lemma}\label{lem:laplace-ratio-n-section}
Fix \(b>a>0\). Then
\[
n\rightarrow R_n(a,b):=\frac{\mathcal L[\widetilde f_n](b)}{\mathcal L[\widetilde f_n](a)}
\]
is strictly increasing.
\end{lemma}

\begin{proof}
For \(x>0\), define the probability measure \(Q_{n,x}\) on \((0,\infty)\) by
\[
dQ_{n,x}(u)=\frac{\widetilde f_n(u)e^{-xu}}{\mathcal L[\widetilde f_n](x)}\,du.
\]
Let
\(
m_n(x)=\mathbb E_{Q_{n,x}}[u].
\)
Since
\[
\frac{\partial}{\partial x}\log \mathcal L[\widetilde f_n](x)
=\frac{\frac{\partial}{\partial x}\mathcal L[\widetilde f_n](x)}{ \mathcal L[\widetilde f_n](x)}
=-\frac{\int_0^\infty u\,\widetilde f_n(u)e^{-xu}\,du}{ \mathcal L[\widetilde f_n](x)}
=-\int_0^\infty u\,dQ_{n,x}(u)
=-m_n(x),
\]
integrating this identity leads to
$$
\log R_n(a,b)=\log \mathcal L[\widetilde f_n](b)-\log \mathcal L[\widetilde f_n](a)=-\int_a^b m_n(x)\,dx.
$$
Now, differentiating with respect to $n$, we obtain
\[
\frac{\partial}{\partial n}\log R_n(a,b)
=
-\int_a^b \frac{\partial}{\partial n}m_n(x)\,dx.
\]
On the other hand, by \eqref{eq7}
\[
\frac{\partial}{\partial n} \widetilde f_n(u)=g_\alpha(u/n)\widetilde f_n(u),
\]
which yields
\begin{align*}
\frac{\partial}{\partial n}m_n(x)&=\frac{\partial}{\partial n} \left(\frac{\int_0^\infty u\,\widetilde f_n(u)e^{-xu}\,du}{ \int_0^\infty\widetilde f_n(u)e^{-xu}\,du} \right)\\
&=  \frac{\int_0^\infty u\,\frac{\partial}{\partial n}\widetilde f_n(u)e^{-xu}\,du}{ \int_0^\infty\widetilde f_n(u)e^{-xu}\,du}-\frac{\int_0^\infty u\,\widetilde f_n(u)e^{-xu}\,du\times \int_0^\infty \frac{\partial}{\partial n}\widetilde f_n(u)e^{-xu}\,du}{ (\int_0^\infty\widetilde f_n(u)e^{-xu}\,du)^2}\\
&=\int_0^\infty u\,g_\alpha(u/n)\,dQ_{n,x}(u)-\int_0^\infty u\,dQ_{n,x}(u)\times \int_0^\infty g_\alpha(u/n)\,dQ_{n,x}(u)\\
&=\Cov_{Q_{n,x}}\!\bigl(u,g_\alpha(u/n)\bigr).
\end{align*}
The function \(u\mapsto u\) is strictly increasing, whereas \(u\mapsto g_\alpha(u/n)\) is strictly decreasing by Lemma~\ref{lem:g-negative-n-section}. Therefore, Chebyshev's covariance inequality  gives
\[
\dfrac{\partial}{\partial n}m_n(x)=\Cov_{Q_{n,x}}\!\bigl(u,g_\alpha(u/n)\bigr)<0,
\]
so that
\[
\frac{\partial}{\partial n}\log R_n(a,b)=-\int_a^b \frac{\partial}{\partial n}m_n(x)\,dx>0.
\]
Therefore \(n\rightarrow R_n(a,b)\) is strictly increasing.
\end{proof}

\begin{theorem}
Extend \(\beta_n(\alpha)\) from integer \(n\) to real \(n>0\) via \eqref{eq:Gtilde-Laplace-n-section}. Then,
for every \(\alpha,\,\mu,\,\nu>0\),  \(n\rightarrow\beta_n(\alpha)\) is strictly decreasing on $(0,\infty)$.
In particular, we have
\(
\beta_{n+1}(\alpha)<\beta_n(\alpha),
\)
for every integer $n\ge1$.
\end{theorem}

\begin{proof}
Write 
\begin{equation}
\begin{split}
\beta_n(\alpha)
&=\frac{(1+2\mu)^2\mathcal{L}_n(\alpha,a_1)+\mu(3+2\mu)\mathcal{L}_n(\alpha,a_3)-(1+\mu)(1+2\mu)\mathcal{L}_n(\alpha,a_2)}{-(1+2\mu)\mathcal{L}_n(\alpha,a_1)+\mu(3+2\mu)\mathcal{L}_n(\alpha,a_3)+(1+\mu)(1+2\mu)\mathcal{L}_n(\alpha,a_2)}\\
&=\frac{(1+2\mu)^2+\mu(3+2\mu)R_n(a_1,a_3)-(1+\mu)(1+2\mu)R_n(a_1,a_2)}{\mu(3+2\mu)R_n(a_1,a_3)+(1+\mu)(1+2\mu)R_n(a_1,a_2)-(1+2\mu)}.
\end{split}
\end{equation}
By a simple differentiation, we obtain
\begin{equation}\label{eq8}
\begin{split}
\frac{\partial}{\partial n}\beta_n(\alpha)
&=-\frac{2\mu(1+\mu)(1+2\mu)(3+2\mu)\bigl(1 - R_n(a_1,a_2)\bigr)}{\Big[\mu(3+2\mu)R_n(a_1,a_3)+(1+\mu)(1+2\mu)R_n(a_1,a_2)-(1+2\mu)\Big]^2}\times\frac{\partial}{\partial n}(R_n(a_1,a_3))\\
&\quad-\frac{2\mu (1+\mu)(1+2\mu)\Bigl((1+2\mu) + (3+2\mu)R_n(a_1,a_3)\Bigr)}{\Big[\mu(3+2\mu)R_n(a_1,a_3)+(1+\mu)(1+2\mu)R_n(a_1,a_2)-(1+2\mu)\Big]^2}\times \frac{\partial}{\partial n}(R_n(a_1,a_2)).
\end{split}
\end{equation}
Since \(a_2>a_1\) and \(a_3>a_1\), Lemma~\ref{lem:laplace-ratio-n-section} implies that $\dfrac{\partial}{\partial n}(R_n(a_1,a_j))>0$ for $j=2,3$. On the other hand, as \(x\mapsto \mathcal{L}_n(\alpha,x)\) is strictly decreasing, we have
$$
0<R_n(a_1,a_j)=\frac{\mathcal{L}_n(\alpha,a_j)}{\mathcal{L}_n(\alpha,a_1)}<1,
$$
for $j=2,3$.
Combining these inequalities in \eqref{eq8} shows that $\dfrac{\partial}{\partial n}\beta_n(\alpha)<0$.
\end{proof}

\section{\label{AppendixD}Asymptotic Analysis}


\subsection{\label{AppendixD-part3}Proof of Eq. \eqref{eq:beta-infty-alpha-section}}
\begin{lemma}
We have the asymptotic expansion 
\begin{equation}
\mathcal{L}_n(\alpha, x) \sim \frac{1}{n(x + c_\alpha)},\qquad\text{for}\quad n\rightarrow\infty.
\end{equation}
\end{lemma}
\begin{proof}
Fix \(x>0\). The generalized Laplace transform is given by
\begin{equation*}
\mathcal{L}_n(\alpha, x) = \int_0^\infty e^{-nx\sigma} [\Phi_\alpha(\sigma)]^n e^{-n\sigma} \, d\sigma = \int_0^\infty e^{n h(\sigma)} \, d\sigma,
\end{equation*}
where
\begin{equation*}
h(\sigma) = \ln\Phi_\alpha(\sigma)- (x+1)\sigma.
\end{equation*}
Since \(e^{-\alpha|m|}<1\) for every \(m\neq0\), we have
\[
\Phi_\alpha(\sigma)=\sum_{m=-\infty}^{\infty}e^{-\alpha|m|}I_{|m|}(\sigma)
<
\sum_{m=-\infty}^{\infty}I_{|m|}(\sigma)=e^\sigma.
\]
Hence \(h(0)=0\) and
\[
h(\sigma)=\log\!\bigl(\Phi_\alpha(\sigma)e^{-\sigma}\bigr)-x\sigma<-x\sigma<0,
\]
for \(\sigma>0\).
Therefore \(\sigma=0\) is the unique global maximizer of \(h\) on \([0,\infty)\).

Using the small-argument expansion of the modified Bessel functions, we have
\[
\Phi_\alpha(\sigma) = 1 + e^{-\alpha}\sigma + O(\sigma^2).
\]
Substituting this into \(h(\sigma)\) and expanding the logarithm, we obtain
\[
h(\sigma) = (e^{-\alpha}\sigma + O(\sigma^2)) - x\sigma = -(x+c_\alpha)\sigma + O(\sigma^2),
\]
where \(c_\alpha = 1 - e^{-\alpha}\).

Fix \(\varepsilon\in(0,x+c_\alpha)\). By the expansion above, there exists \(\delta>0\) such that
\[
-(x+c_\alpha+\varepsilon)\sigma\le h(\sigma)\le -(x+c_\alpha-\varepsilon)\sigma,
\]
for \(0\le \sigma\le \delta\).
It follows that
\[
\int_0^\delta e^{-n(x+c_\alpha+\varepsilon)\sigma}\,d\sigma
\le
\int_0^\delta e^{nh(\sigma)}\,d\sigma
\le
\int_0^\delta e^{-n(x+c_\alpha-\varepsilon)\sigma}\,d\sigma.
\]
For the tail, the bound \(h(\sigma)<-x\sigma\) yields
\[\int_\delta^\infty e^{nh(\sigma)}\,d\sigma
\le
\int_\delta^\infty e^{-nx\sigma}\,d\sigma
=
\frac{e^{-nx\delta}}{nx}.
\]
Combining the last two estimates gives
\[
\frac{1-e^{-n(x+c_\alpha+\varepsilon)\delta}}{n(x+c_\alpha+\varepsilon)}
\le
\mathcal{L}_n(\alpha,x)
\le
\frac{1-e^{-n(x+c_\alpha-\varepsilon)\delta}}{n(x+c_\alpha-\varepsilon)}+\frac{e^{-nx\delta}}{nx}.
\]
Multiplying by \(n\) and letting \(n\to\infty\), we obtain
\[
\frac{1}{x+c_\alpha+\varepsilon}
\le
\liminf_{n\to\infty} n\mathcal{L}_n(\alpha,x)
\le
\limsup_{n\to\infty} n\mathcal{L}_n(\alpha,x)
\le
\frac{1}{x+c_\alpha-\varepsilon}.
\]
Finally, letting \(\varepsilon\downarrow0\) yields the desired asymptotic formula.
\end{proof}
Evaluating this asymptotic expression for $x=\dfrac{1}{2\nu},\,\dfrac{1+2\mu}{2\nu},\,\dfrac{3+2\mu}{2\nu}$ and substituting the results into Eq. \eqref{eq12}, we obtain that
\begin{equation}
\beta_{n}(\alpha)\sim \frac{4c_\alpha^2\nu^2 + (8\mu+6)c_\alpha\nu + 4\mu^2 + 8\mu + 3}{4c_\alpha\nu(c_\alpha\nu + \mu + 1)},\qquad\text{for}\quad n\rightarrow\infty.
\end{equation}
\subsection{\label{AppendixD-part1}Proof of Eq. \eqref{eq:Gtilde-low-nu-section}}

Note the expansion
$$I_m(\sigma) = \frac{(\sigma/2)^m}{m!}  + O(\sigma^{m+2}),\quad\text{for}\quad\sigma\rightarrow0^+,$$
for any $m\in\mathbb{N}$. Retaining terms up to $\sigma$ and substituting this into Eq. \eqref{eq:Phi-alpha-def-self} yields 
\begin{equation}
\begin{split}
\Phi_\alpha(\sigma) 
&=I_0(\sigma)+2\sum_{m=1}^{\infty}e^{-\alpha m}I_m(\sigma)=1 + e^{-\alpha}\sigma + O(\sigma^2).
\end{split}
\end{equation}
On the other hand, note the expansions $(1 + x)^n = 1 + nx + O(x^2)$ and $e^x = 1 + x + O(x^2)$, for small $x$. Then, we have
\begin{align*}
[\Phi_\alpha(\sigma)]^n e^{-n\sigma} &= \left(1 + n e^{-\alpha}\sigma + O(\sigma^2)\right)\left(1 - n\sigma + O(\sigma^2)\right) \\
&= 1-nc_\alpha \sigma+O(\sigma^2),
\end{align*}
for $\sigma\rightarrow0^+$, where $c_\alpha=1-e^{-\alpha}$.
Hence, as \(x\to\infty\), we obtain the expansion
\begin{align*}
\mathcal{L}_n(\alpha,x)
&= \int_0^\infty  [\Phi_\alpha(\sigma)]^n e^{-n\sigma} e^{-nx\sigma} d\sigma\\
&= \frac{1}{x}\int_0^\infty  [\Phi_\alpha(u/x)]^n e^{-nu/x} e^{-nu} du\\
&=\frac{1}{x}\int_0^\infty \Bigl(1-nc_\alpha \frac{u}{x}+O((u/x)^2)\Bigr)\,e^{-nu}\,du\\
&=
\frac{1}{nx}-\frac{c_\alpha}{nx^2}+O(x^{-3}).
\end{align*}
\subsection{\label{AppendixD-part2}Proof of  Eqs. \eqref{eq14} and \eqref{eq:beta-high-nu-n1-section}}
In this section, we will prove Eqs. \eqref{eq14} and \eqref{eq:beta-high-nu-n1-section}.
To establish the asymptotic behavior of $\mathcal{L}_n(\alpha,x)$ as $x\rightarrow0^+$, we first need to study the asymptotic behavior of the phenotypic function $\Phi_\alpha(\sigma)$ for large values of $\sigma$.
\begin{lemma}\label{lem:Phi-large-sigma}
For any fixed discrimination parameter $\alpha > 0$, we have
\[
\Phi_\alpha(\sigma)
\sim
\coth\!\left(\frac{\alpha}{2}\right)\frac{e^\sigma}{\sqrt{2\pi\sigma}}
\qquad\text{as}\quad\sigma\to\infty.
\]
\end{lemma}

\begin{proof}
Let \(r=e^{-\alpha}\in(0,1)\). From the Poisson-kernel representation, write 
\[
\Phi_\alpha(\sigma)
=
\frac{1-r^2}{\pi}e^\sigma
\int_0^\pi \frac{e^{-\sigma(1-\cos\theta)}}{q_r(\theta)}\,d\theta,
\]
where $q_r(\theta):=1-2r\cos\theta+r^2$.
Note that \(q_r(\theta)\geq q_r(0)=(1-r)^2>0\) for $\theta\in[0,\pi]$. To prove the desired estimate, it suffices to estimate
\[
J_\sigma:=\int_0^\pi \frac{e^{-\sigma(1-\cos\theta)}}{q_r(\theta)}\,d\theta .
\]
As $\lim_{\theta\rightarrow0^{+}}\dfrac{1-\cos\theta}{\theta^2}=1/2$, then there exists
\(\delta\in(0,\pi)\), such that
\begin{equation}\label{eq34}
1-\cos\theta\ge\theta^2/4 \qquad\text{for}\quad \theta\in[0,\delta].
\end{equation}
Split the integral $J_\sigma=J^{(1)}_\sigma+J^{(2)}_\sigma$, where 
$J^{(1)}_\sigma=\int_0^\delta \frac{e^{-\sigma(1-\cos\theta)}}{q_r(\theta)}\,d\theta$ and 
$J^{(2)}_\sigma=\int_\delta^\pi \frac{e^{-\sigma(1-\cos\theta)}}{q_r(\theta)}\,d\theta$.

For the first integral, with the change of variables \(u=\sqrt{\sigma}\,\theta\), we have
\[
\sqrt{\sigma}\,J^{(1)}_\sigma
=
\int_0^{\delta\sqrt{\sigma}}
\frac{\exp\!\bigl(-\sigma(1-\cos(u/\sqrt{\sigma}))\bigr)}
{q_r(u/\sqrt{\sigma})}\,du.
\]
For each fixed \(u\ge0\), we have
\begin{equation}
\lim_{\sigma\rightarrow\infty}\frac{\exp\!\bigl(-\sigma(1-\cos(u/\sqrt{\sigma}))\bigr)}
{q_r(u/\sqrt{\sigma})}=\frac{e^{-u^2/2}}{(1-r)^2}.
\end{equation}
By Eq. \eqref{eq34}, we have 
\[
0\le
\frac{\exp\!\bigl(-\sigma(1-\cos(u/\sqrt{\sigma}))\bigr)}
{q_r(u/\sqrt{\sigma})}
\le
\frac{e^{-u^2/4}}{(1-r)^2}
\qquad\text{whenever }u\le\delta\sqrt{\sigma}.
\]
Using the dominated convergence theorem yields
\[
\lim_{\sigma\rightarrow\infty}\sqrt{\sigma}\,J^{(1)}_\sigma=\frac{1}{(1-r)^2}\int_0^\infty e^{-u^2/2}\,du=\frac{1}{(1-r)^2}\sqrt{\frac{\pi}{2}},
\]
which shows that
\begin{equation}\label{eq35}
J^{(1)}_\sigma\sim \frac{1}{(1-r)^2}\sqrt{\frac{\pi}{2\sigma}}
\end{equation}
as $\sigma\rightarrow\infty$.

For the second integral $J^{(2)}_\sigma$, note that
\[1-\cos\theta\ge 1-\cos\delta>0,\quad\text{for}\quad\theta\in[\delta,\pi].\]
Therefore, we have
\begin{equation}\label{eq36}
J^{(2)}_\sigma=\int_\delta^\pi \frac{e^{-\sigma(1-\cos\theta)}}{q_r(\theta)}\,d\theta
=
O\!\left(e^{-\sigma(1-\cos\delta)}\right)=o\left(\sigma^{-1/2}\right),
\end{equation}

 Combining \eqref{eq35} and \eqref{eq36} leads to 
\[
J_\sigma\sim\frac{1}{(1-r)^2}\sqrt{\frac{\pi}{2\sigma}}.
\]
Substituting this result back into the expression for $\Phi_\alpha(\sigma)$ yields
\[
\Phi_\alpha(\sigma)\sim\frac{1-r^2}{\pi}e^\sigma\frac{1}{(1-r)^2}\sqrt{\frac{\pi}{2\sigma}}=\coth\!\left(\frac{\alpha}{2}\right)\frac{e^\sigma}{\sqrt{2\pi\sigma}}.
\]
\end{proof}
Set $C_\alpha:=\coth\!\left(\frac{\alpha}{2}\right)$ and $A_n(\sigma):=[\Phi_\alpha(\sigma)]^n e^{-n\sigma}$.
By Lemma~\ref{lem:Phi-large-sigma}, we have
\[
A_n(\sigma)\sim \frac{C_\alpha^n}{(2\pi\sigma)^{n/2}},
\]
when \(\sigma\to\infty\). We now analyze
\[
\mathcal{L}_n(\alpha,x)=\int_0^\infty A_n(\sigma)e^{-nx\sigma}\,d\sigma.
\]

When \(n=1\), fix \(\varepsilon\in(0,1)\). There exists \(S>0\) such that for every \(\sigma\ge S\), 
\begin{equation}
(1-\varepsilon)\frac{C_\alpha}{\sqrt{2\pi\sigma}}\le A_1(\sigma)\le(1+\varepsilon)\frac{C_\alpha}{\sqrt{2\pi\sigma}}.
\end{equation}
Then, we have
\begin{equation}\label{eq32}
\frac{(1-\varepsilon)C_\alpha}{\sqrt{2\pi}}\int_{S}^{\infty}\sigma^{-1/2} e^{-x\sigma}d\sigma\le \int_{S}^{\infty}A_1(\sigma)e^{-x\sigma}d\sigma\le\frac{(1+\varepsilon)C_\alpha}{\sqrt{2\pi}}\int_{S}^{\infty} \sigma^{-1/2}e^{-x\sigma}d\sigma.
\end{equation}
Using the change of variable \(t=x\sigma\), we obtain
\[
\int_S^\infty e^{-x\sigma}\sigma^{-1/2}\,d\sigma=\frac{1}{\sqrt{x}}\int_{xS}^\infty e^{-t}t^{-1/2}\,dt\sim\frac{\Gamma(1/2)}{\sqrt{x}}=\frac{\sqrt{\pi}}{\sqrt{x}}, \text{\;as\;}x\rightarrow0^+.
\]
Then, there exists $\delta>0$, where 
\begin{equation}\label{eq33}
(1-\varepsilon)\frac{\sqrt{\pi}}{\sqrt{x}}<\int_S^\infty e^{-x\sigma}\sigma^{-1/2}\,d\sigma< (1+\varepsilon)\frac{\sqrt{\pi}}{\sqrt{x}},
\end{equation}
for any $0<x<\delta$. Combining \eqref{eq32} and \eqref{eq33}, we get
\begin{equation}
\frac{(1-\varepsilon)^2C_\alpha}{\sqrt{2x}}\le \int_{S}^{\infty}A_1(\sigma)e^{-x\sigma}d\sigma\le\frac{(1+\varepsilon)^2C_\alpha}{\sqrt{2x}}, \text{\; for\;}0<x<\delta.
\end{equation}
Since \(\int_0^S A_1(\sigma)e^{-x\sigma}\,d\sigma=O(1)=o(x^{-1/2})\), the contribution from \([0,S]\) is negligible. 
Therefore, we obtain
\[
\mathcal{L}_1(\alpha,x)\sim \frac{C_\alpha}{\sqrt{2x}}.
\]
Evaluating this asymptotic expression for $x=\dfrac{1}{2\nu},\,\dfrac{1+2\mu}{2\nu},\,\dfrac{3+2\mu}{2\nu}$ and substituting the results into Eq. \eqref{eq12}, we obtain that
\begin{equation}
\beta_1(\alpha) \approx \frac{(1+2\mu)^2 + \mu\sqrt{3+2\mu} - (1+\mu)\sqrt{1+2\mu}}{\mu\sqrt{3+2\mu} + (1+\mu)\sqrt{1+2\mu} - (1+2\mu)}, \qquad\text{for}\quad\nu\rightarrow\infty.
\end{equation}

Similarly, we have
\begin{subequations}
\begin{align}
\mathcal{L}_2(\alpha,x)&\sim\frac{C_\alpha^2}{2\pi}\log\!\left(\frac{1}{x}\right),\\
\mathcal{L}_n(\alpha, x) &\sim \int_0^\infty [\Phi_\alpha(\sigma)]^n e^{-n\sigma}\, d\sigma \in (0,\infty),
\end{align}
\end{subequations}
for $n\geq3$. This will lead to 
\begin{equation}
\lim_{\nu\to\infty}\beta_n(\alpha)=1,
\end{equation}
for $n\geq2$.
\bibliographystyle{unsrt}
\bibliography{bibliography}

\end{document}